\documentclass[conference]{IEEEtran}

\usepackage{amsmath}
\usepackage{amsthm}
\usepackage{amsfonts,amssymb}
\usepackage{graphicx}
\usepackage{algpseudocode}
\usepackage{algorithm}
\usepackage{subcaption}
\usepackage{bm}
\usepackage{xcolor}
\usepackage{float}
\usepackage{xspace}
\usepackage{multirow}

\usepackage{enumitem}
\usepackage{stfloats}
\usepackage{soul}
\usepackage[skins]{tcolorbox}

\newtheorem{definition}{Definition}
\newtheorem{theorem}{Theorem}

\newcommand{\SGR}{\ensuremath{\mathsf{SGR}}\xspace}
\newcommand{\ASG}{\ensuremath{\mathsf{ASG}}\xspace}

\newcommand{\mypara}[1]{\vspace*{1pt}\noindent{\textbf{#1}}\xspace}
\newcommand{\mysubpara}[1]{\vspace{1pt}\noindent{\textit{#1}}\xspace}

\begin{document}
 
\title{On the Robustness of LDP Protocols for Numerical Attributes under Data Poisoning Attacks}

\author{\IEEEauthorblockN{Xiaoguang Li}
\IEEEauthorblockA{Xidian University, Purdue University\\
xg\_li@outlook.com}
\and
\IEEEauthorblockN{Zitao Li}
\IEEEauthorblockA{Alibaba Group (U.S.) Inc.\\
zitao.l@alibaba-inc.com}
\and
\IEEEauthorblockN{Ninghui Li}
\IEEEauthorblockA{Purdue University\\
ninghui@purdue.edu}
\and
\IEEEauthorblockN{Wenhai Sun}
\IEEEauthorblockA{Purdue University\\
whsun@purdue.edu}
}



\maketitle

\begin{abstract}
Recent studies reveal that local differential privacy (LDP) protocols are vulnerable to data poisoning attacks where an attacker can manipulate the final estimate on the server by leveraging the characteristics of LDP and sending carefully crafted data from a small fraction of controlled local clients. This vulnerability raises concerns regarding the robustness and reliability of LDP in hostile environments. 

In this paper, we conduct a systematic investigation of the robustness of state-of-the-art LDP protocols for numerical attributes, i.e.,  categorical frequency oracles (CFOs) with binning and consistency, and distribution reconstruction. We evaluate protocol robustness through an attack-driven approach and propose new metrics for cross-protocol attack gain measurement. The results indicate that Square Wave and CFO-based protocols in the \textit{Server} setting are more robust against the attack compared to the CFO-based protocols in the \textit{User} setting. Our evaluation also unfolds new relationships between LDP security and its inherent design choices. We found that the hash domain size in local-hashing-based LDP has a profound impact on protocol robustness beyond the well-known effect on utility. Further, we propose a \textit{zero-shot attack detection} by leveraging the rich reconstructed distribution information. The experiment show that our detection significantly improves the existing methods and effectively identifies data manipulation in challenging scenarios. 
\end{abstract}

\section{Introduction}

Local differential privacy is the \textit{de facto} standard for privacy-protective data collection and statistical estimates \cite{duchi2013local, wang2017locally, li2020estimating, imola2021locally} against untrusted data curators. It has been deployed by many tech companies, such as Google \cite{erlingsson2014rappor}, Apple \cite{apple}, and Microsoft \cite{ding2017collecting}. Recent studies raise security concerns about LDP since malicious local data providers can report bogus data and take advantage of the LDP design to effectively skew the estimates on the server. This data poisoning attack has been confirmed by prior works in various statistic tasks, such as frequency and heavy hitter \cite{cao2021data, cheu2021manipulation}, key-value data \cite{wu2021poisoning}, mean and variance \cite{li2022fine}. These studies focused on attack exploration, which is important for us to understand the attack strategies and their impacts on the adoption of LDP in practice. However, there is still a deficiency in understanding the robustness and reliability of LDP with various design principles and the defensive implications due to their algorithmic differences. Specifically, we aim to answer the following research questions in this paper.

\noindent\textbf{RQ-1.} \mysubpara{How can we translate attack effectiveness to LDP robustness given diverse LDP designs to enable meaningful comparison?} We would like to understand the principle of various attack strategies and propose unified metrics to quantify the attack strength for the foundation of fair LDP robustness comparison.

\noindent\textbf{RQ-2.} \mysubpara{How robust are the compared LDP protocols against the data poisoning attack?}  One protocol is less robust than another if, under the same setting, the attack is more effective in influencing the estimated statistical quantity. Currently, LDP protocols are often compared based on versatility (i.e., types of statistics that can be estimated) and estimation accuracy/utility; we argue that robustness adds another critical dimension and would generate new insights into security-aware LDP design. 

\noindent\textbf{RQ-3.} \mysubpara{Can we design an effective detection function with minimum information for existing vulnerable LDP protocols to elevate their attack resilience?} A unified solution facilitates adoption and standardization of LDP protocols but is also challenging since the data ground truth is unknown and clients in different protocols report distinct information.

In this paper, we look into the numerical domain, which includes various applications utilizing ordinal or numerical data, e.g., income, age, health assessment data, etc. We focus on the distribution estimation as it is a primary statistic from which other statistical properties can be computed without additional privacy costs. Though the attack design is an essential part of our evaluation framework, prior research on attacking statistical sub-tasks (e.g. mean/variance) \cite{li2022fine, cheu2021manipulation} provides little insight regarding vulnerability exploit due to different estimation tasks and attack goals. 

In particular, we look at two types of distribution estimation techniques: \textit{categorical frequency oracles with binning and consistency}, and \textit{direct distribution recovery}. The protocols in the first category bin the numerical attributes and then apply CFOs to estimate the histogram. They then apply post-processing to force the histogram to be consistent with natural distributions. In this category, we study four state-of-the-art CFOs, i.e., \textbf{HST} \cite{cheu2021manipulation}, \textbf{GRR}, \textbf{OUE} and \textbf{OLH} \cite{wang2017locally}. Depending on who (either the server or the user) selects the hash functions, we further look at HST and OLH in the \textit{Server} setting and \textit{User} setting respectively. For the direct distribution recovery, we study a representative mechanism, \textbf{SW}~\cite{li2020estimating}. In SW, the server considers the ordinal nature of the numerical domain and uses Expectation Maximization with Smoothing to reconstruct the distribution.

We propose an attack-driven evaluation framework to measure the robustness of each attacked protocol. To this end, we present a suite of \textbf{Distribution-Shift Attacks} in which the attacker aims to shift the estimated distribution to the right-most end.\footnote{Our attacks can be adapted to other regions.} We propose new metrics to quantitively compare robustness across different LDP protocols.  The first is \textbf{Absolute Shift Gain} (\ASG) inspired by Wasserstein distance to quantify the extent to which the distribution is shifted towards the right side. Smaller \ASG indicates better attack resilience. We also look at \textbf{Shift Gain Ratio} (\SGR), the ratio of \ASG of the proposed attack to that of a baseline attack (i.e., sending the maximum input value $x_{\max}$ to skew the distribution).  The baseline attack is indistinguishable from the case where the true inputs of a group of users are also $x_{\max}$, and thus its impact is the same on any LDP mechanism that produces an unbiased estimation. Therefore, \SGR helps us evaluate the protocol robustness using the baseline attack as a reference regardless of the attack background, such as the number of malicious users and the ground truth distribution. $\SGR > 1$ indicates the attack outperforms the baseline. The larger the metric is, the less robust the underlying LDP protocol is.

\textbf{Results on Robustness.} First, the robustness of the studied protocols varies under the attack. For the attack strategy that maximizes the frequency of the right-most bin in the domain, SW is more robust than CFO-based mechanisms because of the smoothing step of SW. To further reveal the more vulnerable sides of SW, we allow the attacker to report values sampled in a higher-end range. We observe that SW and CFOs in the \textit{Server} setting consistently outperform CFOs in the \textit{User} setting. Our experiment reports enhanced robustness of all protocols with an increasing privacy budget. The attack effectiveness will degrade to the baseline with a reasonably large $\epsilon$ (e.g., 1).
In addition, our study uncovers a new correlation between protocol design and LDP security beyond prior findings that only privacy preference affects the LDP robustness \cite{cao2021data, li2022fine}. Our evaluation shows a trade-off between the security and hash domain size of local-hashing-based LDP, i.e., the larger the hash domain size, the less robust the protocol.

\textbf{Detection Exploration.} To help vulnerable protocols restore the attack resilience, we design a \textit{zero-shot attack detection} by analyzing the distribution of reported values. Our detection provides a unified solution that accommodates different LDP protocols. It demonstrates significant performance improvement compared to the state-of-the-art detection \cite{cao2021data}. 
Since the fake values are crafted without following the LDP protocol, the main idea of our detection is to capture the difference between the randomness of fake values and the LDP perturbation. Specifically, we synthesize the data based on the estimated distribution and push them through the LDP perturbations. Since the perturbed results are not polluted, they are expected to carry strong LDP randomness and thus can be used as a benchmark of ``no attack". The noisy report then can be considered unpolluted when it is statistically close to this benchmark.
To evaluate detection performance, we plot the Receiver Operating Characteristic (ROC) curve and measure the Area Under the Curve (AUC), where a larger AUC indicates a better detection result. The experimental results on both real-world and synthesized datasets show that our detection substantially outperforms the state-of-the-art \cite{cao2021data}. In contrast to the existing method that fails to identify the attack in most cases, our detection demonstrates superior performance, especially on SW, OUE, and CFO-based mechanisms under \textit{User} setting. Specifically, the AUC of our detection constantly remains over 0.92 with the attacker only controlling 5\% of the total clients. When the attack appears to evade our detection, its efficacy also diminishes quickly to the baseline, indicating that it is no longer a serious security concern. 
We summarize our contributions below.

\begin{enumerate}[leftmargin=*, topsep=0pt, parsep=0pt]

    \item We systematically investigate the robustness of the LDP protocols for numerical data against data poisoning attacks. The studied protocols consist of a wide array of state-of-the-art distribution estimation mechanisms to show the current threat landscape and further reveal its profound impact.
    
    \item To evaluate the robustness of the protocols, we look into the design of effective attacks and provide new metrics to measure LDP robustness informed by attack strength regardless of diverse design principles of underlying LDP protocols.
   
    \item We both theoretically and empirically evaluate and compare the robustness of the studied LDP protocols. We found that these protocols are not equally resilient against the threat. The CFOs under \textit{Server} setting and SW offer better security. We also uncover new factors in LDP design that influence protocol robustness other than the well-known privacy-security relationships.
    
    \item We propose a novel detection method to identify data manipulation. Our method leverages the distribution properties of the reported data without knowing the ground truth information to detect attacks that introduce statistically non-negligible bias to the distribution. The experiment shows that our method consistently outperforms the existing ones in many challenging scenarios.  
\end{enumerate}

\section{Background and Related Work} \label{sec:background}
\subsection{Local Differential Privacy}
LDP considers the setting that there are $n$ local users and a remote untrusted data collector.  
Each user possesses private data $x \in \mathcal{D}$, which is of interest to the data collector. 
To protect privacy, each user randomly perturbs $x$ with an algorithm $\Psi$, and only reports the perturbed data $\Psi(x)$ to the data collector. The algorithm $\Psi$ provides LDP protection if and only if it satisfies the following definition.
\begin{definition}[$\epsilon$-Local Differential Privacy \cite{duchi2013local}]
    An algorithm $\Psi(\cdot): \mathcal{D} \rightarrow \hat{\mathcal{D}}$ satisfies $\epsilon$-LDP if and only if for any $x_i, x_j \in \mathcal{D}$ and any $t \in \hat{\mathcal{D}}$, the following inequality holds,
    \begin{align*}
        \Pr[\Psi(x_i) = t] \leq e^\epsilon \Pr[\Psi(x_j) = t] .
    \end{align*}
\end{definition}
A smaller $\epsilon$, referred to as the \textit{privacy budget}, corresponds to a more tightened privacy level but also a lower data utility.

\subsection{CFO with binning and consistency}

We first introduce the CFOs used in this paper and then discuss the post-processing methods for consistency. 

\mypara{Categorical frequency oracles.} 
We assume the protocol discretizes the numerical domain [0, 1] into $m_o$ bins, and denote the index of the bin that contains the user's private value $x$ as $x_b$ for convenience.

\mysubpara{Generalized Randomized Response (GRR) \cite{wang2017locally}.}
GRR directly perturbs each input $x_b$ by keeping $x_b$ unchanged with probability $p = \frac{e^\epsilon}{e^\epsilon + m_o - 1}$ and changing it to another index with probability $q = \frac{1}{e^\epsilon + m_o - 1}$. The frequency of each bin $B_i (\forall i \in [m_o])$ is estimated by aggregation function $\Phi_{\mathrm{GRR}}(B_i) = \frac{\sum_{j=1}^{n} \mathbb{I}_{[i]}(y^{(j)}) - nq}{n(p-q)}$, where $y^{(j)}$ is the $j$-th user's report and $\mathbb{I}_{[i]}(y^{(j)})$ is the indicator function equal to 1 if $y^{(j)} = i$ and 0 otherwise.

\mysubpara{Optimal Unary Encoding (OUE) \cite{wang2017locally}.}
OUE is one of the optimal categorical frequency oracles that attains the theoretical lower bound of the LDP protocol $L_2$ error.
It first encodes each input $x_b$ as a one-hot vector $\bm{v} = [0, ..., 1, .., 0]$, where all elements are $0$ except for the element at position $x_b$. 
Then it flips the bits in $\bm{v}$ to get $\hat{\bm{v}}$ as follows: if the bit is 1, it is flipped to 0 with probability $\frac{1}{2}$; otherwise a bit 0 is flipped to 1 with probability $\frac{1}{e^\epsilon + 1}$.

The frequency of each bin $B_i (\forall i \in [m_o])$ is estimated by aggregation function $\Phi_{\mathrm{OUE}}(B_i)$. With the $j$-th user's report denoted as $\hat{\bm{v}}^{(j)}$, $\Phi_{\mathrm{OUE}}(B_i) = \frac{\sum_{j=1}^{n} \hat{\bm{v}}^{(j)}[i] - \frac{n}{e^\epsilon + 1}}{n(\frac{1}{2}-\frac{1}{e^\epsilon + 1})}$.


\mysubpara{Optimal Local Hashing (OLH) \cite{wang2017locally}.}
OLH is an optimal categorical frequency oracle with the same minimal $L_2$ error as OUE.
OLH leverages a family of hash functions $\mathbf{H}$, each of which maps $x_b \in \{1, ..., m_o \}$ to $x_h \in \{1, ..., g\}$, where $g = \lfloor e^\epsilon + 1 \rfloor$. For simplicity, we denote the domain $\{1, ..., m_o \}$ and $\{1, ..., g\}$ as $[m_o]$ and $[g]$. An example of a hash function family is \textsf{xxh32} \cite{xxhash} with different seeds.
In OLH, each user first uses a randomly selected hash function $H \in \mathbf{H}$ to encode the value $x_b$ as $x_h = H(x_b)$.
Given the hash function, each user perturbs the hash value $x_h$ as follows and reports the tuple $\langle H, \hat{x}_h \rangle$.
\begin{align*}
    \forall_{\hat{x}_h \in [g]} \Pr[\Psi_{\mathrm{OLH}}(x_b) = \hat{x}_h] =
    \begin{cases}
        p = \frac{e^\epsilon}{e^\epsilon + g - 1}, & \mathrm{if} \,\,\, \hat{x}_h = x_h \\
        q = \frac{1}{e^\epsilon + g - 1}, & \mathrm{if} \,\,\, \hat{x}_h \neq x_h
    \end{cases}
\end{align*}

Then OLH mechanism estimates the frequency of each bin $B_i (\forall i \in [m_o])$. 
Denote the reported tuple from $j$-th user as $y^{(j)} = \langle H^{(j)}, \hat{x}_h^{(j)} \rangle$. It first counts the noisy data falling into $B_i$ as $C(B_i) = |\{ j | H^{(j)}(x_b) = \hat{x}_h^{(j)}, x_b = i \}|$ and transforms it by an aggregation function $\Phi_{\mathrm{OLH}}(B_i) = \frac{C(B_i) - \frac{n}{g}}{n(\frac{e^\epsilon}{e^\epsilon + g - 1} - \frac{1}{g})}$ to obtain the unbiased frequency estimate.

\mysubpara{ExplicitHist (HST) \cite{cheu2021manipulation}.}
Initially, each user $j$ independently samples a uniform public vector $\bm{s}^{(j)} \in \{\pm 1\}^{m_o}$. Then he picks the $x_b$-th element $s^{(j)}_{x_b}$ in the vector and perturbs it as follows to report.
\begin{align*}
    \Pr[\Psi_{HST}(s^{(j)}_{x_b}) = \hat{s}^{(j)}_{x_b}] =
    \begin{cases}
       \frac{e^\epsilon}{e^\epsilon + 1} , & \mathrm{if} \,\,\, \hat{s}^{(j)}_{x_b} = \frac{e^\epsilon + 1}{e^\epsilon - 1} \times s^{(j)}_{x_b} \\
       \frac{1}{e^\epsilon + 1} , & \mathrm{if} \,\,\, \hat{s}^{(j)}_{x_b} = -\frac{e^\epsilon + 1}{e^\epsilon - 1} \times s^{(j)}_{x_b}
    \end{cases}
\end{align*}

The protocol then aggregates all reports and estimates the frequency of each bin $B_i (\forall i \in [m_o])$ by the aggregation function $\Phi_{\mathrm{HST}}(B_i) = \frac{1}{n} \sum_{j=1}^{n} \hat{s}^{(j)}_{x_b} \times \bm{s}^{(j)}[i]$. The effect is that each user with value in $B_i$ contributes $\frac{e^\epsilon + 1}{e^\epsilon - 1}$ with probability $\frac{e^\epsilon}{e^\epsilon + 1}$ and $-\frac{e^\epsilon + 1}{e^\epsilon - 1}$ with probability $\frac{1}{e^\epsilon + 1}$ to $\Phi_{\mathrm{HST}}(B_i)$. Thus the expected contribution is 1.


\mypara{Relationship between HST and OLH.}
Previous work \cite{wang2017locally} shows that HST is equivalent to OLH when $g=2$  since HST essentially uses binary local hashes to map each input into the domain $\{-1, 1\}$.


A key challenge of applying binning is to choose an appropriate binning granularity $m_o$ \cite{li2020estimating}. A larger $m_o$ leads to accumulated noise error. A smaller one may discard part of the distribution information thus introducing a large bias. A good choice of $m_o$ depends on balancing the above two sources of errors given a privacy budget $\epsilon$. For example, some empirical results~\cite{li2020estimating} show that $m_o = 32$ is the best choice for some real-world distributions.

\mypara{Consistency post-processing.}
Various post-processing methods designed for consistency are proposed in previous work \cite{wang2020locally}. Among them, Norm-Sub is the maximum likelihood estimator for noisy estimates and achieves the most accurate post-processed result overall. The idea of Norm-Sub is to convert negative values to zero and add a factor $\alpha$ to the remaining values such that the total frequencies sum up to one, i.e., $\sum_{i=1}^{m} \max{(\hat{h}_i + \alpha, 0)} = 1$. 
Then the post-processed result for $\hat{h}_i$ (i.e., the estimated frequency $\Phi(B_i)$ of the $i$-th bin) is $\tilde{h}_{i} = \max{(\hat{h}_i + \alpha, 0)}$.  We adopt Norm-Sub in this paper to couple with the CFOs.

\subsection{Distribution Reconstruction}
\mysubpara{Square Wave (SW) \cite{li2020estimating}.}
SW mechanism is another state-of-the-art mechanism supporting numerical distribution estimation. Different from binning-based methods that perturb values in a discrete manner, the SW mechanism considers the ordinal information of the numerical domain, and each value is perturbed as 
\begin{align*}
    \Pr[\Psi_{\mathrm{SW}}(x) = \hat{x}] =
    \begin{cases}
        p, & \mathrm{if} \,\,\, \left| x - \hat{x} \right| \leq b \\
        q, & \mathrm{if} \,\,\, \left| x - \hat{x} \right| > b
    \end{cases},
\end{align*}
where $p = \frac{e^\epsilon}{2be^\epsilon + 1}$, $q = \frac{1}{2be^\epsilon + 1}$ and $b = \frac{\epsilon e^\epsilon - e^\epsilon + 1}{2e^\epsilon (e^\epsilon - 1 - \epsilon)}$. While the input domain is $[0, 1]$, the output domain of SW is $[-b, 1+b]$.

The aggregation $\Phi_{\mathrm{SW}}(\cdot)$ of SW is not a closed-form mathematical expression, but the Expectation–Maximization with Smoothing (EMS) algorithm approximates the maximum-likelihood distribution given the reported values while keeping the smoothness of distributions in the numerical domain by weight average. 


SW also needs to choose the number of bins $m_s$.
Similar to CFO, smaller $m_s$ also raise the bias in the estimated result.
However, larger $m_s$ does not lead to larger LDP noise because the binning is in the aggregation phase and keeping increasing $m_s$ does not significantly improve SW performance but makes EMS converge slowly when $m_s$ is large enough.
According to the empirical study \cite{li2020estimating}, SW performs best in most cases under $m_s = 512$.

\subsection{Related Work}
\noindent\textbf{Data Poisoning Attack.}
We discuss the related work on data poisoning attacks to LDP. Cheu \textit{et al.} \cite{cheu2021manipulation} studied the issue of data manipulation and showed that the vulnerability is inherent to non-interactive LDP protocols. The attacks on frequency estimation and key-value data collection were also studied in \cite{cao2021data} and \cite{wu2021poisoning} respectively. 
They focused on maximizing the statistics of the attacker-chosen items by sending bogus data to the server.
In contrast to straightforward maximization, Li \textit{et al.} \cite{li2022fine} considered an attack that allows the attacker to fine-tune the final estimate to a target value for the mean and variance estimation. We, in this paper, also target the numerical data but focus on the more flexible and versatile distribution estimation, where prior work cannot be applied.

\vspace{2pt}
\noindent\textbf{Mitigation.}
To detect the attack attempt, Cao \textit{et al.} \cite{cao2021data} proposed malicious user detection (MUD) and conditional probability-based attack detection (CPAD) against the attack only on frequency estimation. However, tuning parameters of CPAD for good performance depends on the ground-truth knowledge, which may not be practical. MUD was also adopted in \cite{huang2024ldpguard} for attack detection for OUE and OLH. We compare MUD with our detection method for CFO-based protocols in Section~\ref{sec:detection_evaluation}. Wu \textit{et al.} studied detection for attacks on key-value data \cite{wu2021poisoning}. However, it is challenging to obtain true data of specific genuine users to train the classifier. They also designed the detection for interactive LDP protocols by tracking the reports in different rounds of communications between users and the server. Nevertheless, this method cannot work for non-interactive LDP protocols for numerical data in that the distribution is estimated through one-round reports.
Pollution tolerance \cite{sun2024ldprecover, yan2023towards, li2022fine} is another type of defense that serves as a post-attack recovery method to restore the corrupted utility as much as possible. However, due to the lack of ground truth and attack information (e.g., whether the attack has been launched, the fraction of malicious users, etc.), tuning parameters for satisfactory recovery is challenging.
In this paper, we design a novel zero-shot attack detection to identify poisoning attacks without ground truth information.

\section{Robustness Evaluation Framework} \label{sec:evaluation_framework}
In this section, we introduce the evaluation framework to assess the robustness of LDP protocols for numerical data. We first introduce the overview of our evaluation methodology.
Then we concretely describe each step in the framework.

\subsection{Overview}
The main methodology of the framework is to conduct an attack-driven robustness evaluation for the studied LDP protocols. By measuring the attack effectiveness, the resilience of each protocol against the data poisoning attack can be learned and compared. To this end, the first step of the evaluation framework is to formulate the attack model and define metrics for the attack effectiveness measurement. 

A straightforward metric to measure the attack efficacy is the Wasserstein distance, which measures the distance between the estimated distribution under attack and the clean one. However, Wasserstein distance is directionless and it cannot reflect in which direction the attacked distribution is skewed. Therefore, it is necessary to design new metrics to overcome the limitations and facilitate attack evaluation.

In light of the distinct design choices of the underlying LDP protocols, one type of attack may not stay effective on all of them, leading to biased evaluation results. Hence, we customize the attack to make sure that the attack goal can be achieved and the result can be fairly compared.

Given the attacks and metrics, we will empirically study the robustness of the LDP protocols in various settings. It is anticipated that lower attack effectiveness indicates better LDP robustness. In what follows, we introduce each step in the evaluation framework and the experimental results in Section~\ref{sec:attack_experiments}.


\subsection{Attack Model}

\noindent\textbf{Attacker's Capability.}
To be consistent with previous work \cite{cheu2021manipulation}, we assume the attacker can control $n_f = \beta n$ users to send crafted fake values $\hat{Y} = [\hat{y}^{(i)}]_{i=1}^{n_f}$ to the server, where $n$ is the total number of users and $0 \leq \beta \leq 1$. 
The remaining $n_g = (1-\beta)n$ users are considered benign, sending the perturbed true values to the server.
We also assume the attacker can access the encoding and perturbation steps in the LDP protocol because these steps are deployed locally. Therefore, the attacker knows the parameters of the LDP protocols, including the privacy budget $\epsilon$, the perturbation probability, and the number of bins.

\noindent\textbf{Attacker's Goal.}
The attacker aims to shift the distribution to the right-most end as much as possible. Our attacks are adaptable to other regions of the distribution. For ease of demonstration, we focus on the right side in this paper.

Despite intuitive, our attack poses real-world threats. For example, the attacker can trick consumers into downloading target apps in online stores by skewing their rating distribution to the higher end. There may exist more sophisticated attacks with other goals. However, they may require additional knowledge, such as the fraction of controlled users $\beta$ and the true data distribution, which is difficult to acquire in practice. On the other hand, our attack is inspired by prior work \cite{cao2021data} and generalizable by focusing on the best-effort maximization for distribution skewness, which enables robustness analysis of LDP against known attack strategies.


\vspace{2pt}
\noindent\textbf{Baseline Attack.}
In a general LDP context, the baseline attack on the LDP protocol is that the attacker skews the estimates by providing fake values in the input domain to the LDP protocol and reporting the corresponding randomized outputs. The baseline is considered a universal attack since its effect is indistinguishable from that of honest protocol execution. It was also adopted and called \textit{input manipulation/attack} in prior work \cite{cheu2021manipulation, li2022fine} for attack performance evaluation. For distribution-shift attacks, since the frequencies of all bins sum up to one, an effective way to shift the distribution is to maximally increase the frequency of the right-most bin and decrease the frequencies of other bins at the same time. Therefore, the baseline attacker injects maximums in the input domain of the LDP protocol as fake values. 
Specifically, the baseline attacker repeats $\Psi_{\mathrm{HST}}(m_o)$ (or correspondingly, $\Psi_{\mathrm{OLH}}(m_o)$ for OLH, $\Psi_{\mathrm{SW}}(1)$ for SW) for $n_f$ times and sends the randomized outputs to the server.

The users controlled by a baseline attacker behave identically to those users who happen to own and submit the same values by following the protocol. Therefore, without prior information about the true distribution or other authentication techniques, the server has no way to detect the baseline attack.
Due to the unbiasedness of LDP protocols, the aggregated distribution (before consistency) of the baseline attack $\hat{X}_{a}^{base}$ is approximately equal to $\mathbb{E}(\hat{X}_{a}^{base}) = X_a^{base}$, where $X_a^{base}$ is the skewed distribution in the input domain after baseline attack. Therefore, the outcome of the baseline attack is almost independent of post-processing since $X_a^{base}$ has satisfied consistency.

\subsection{Attack Effectiveness}
\label{sec:attack_effectiveness}
To quantitatively measure the efficacy of the attacks, we propose two metrics in this paper to show how far the poisoned distribution can deviate from the true distribution.

\vspace{2pt}
\noindent\textbf{Absolute Shift Gain (\ASG).}
\begin{figure}
    \centering
    \includegraphics[scale = 0.28]{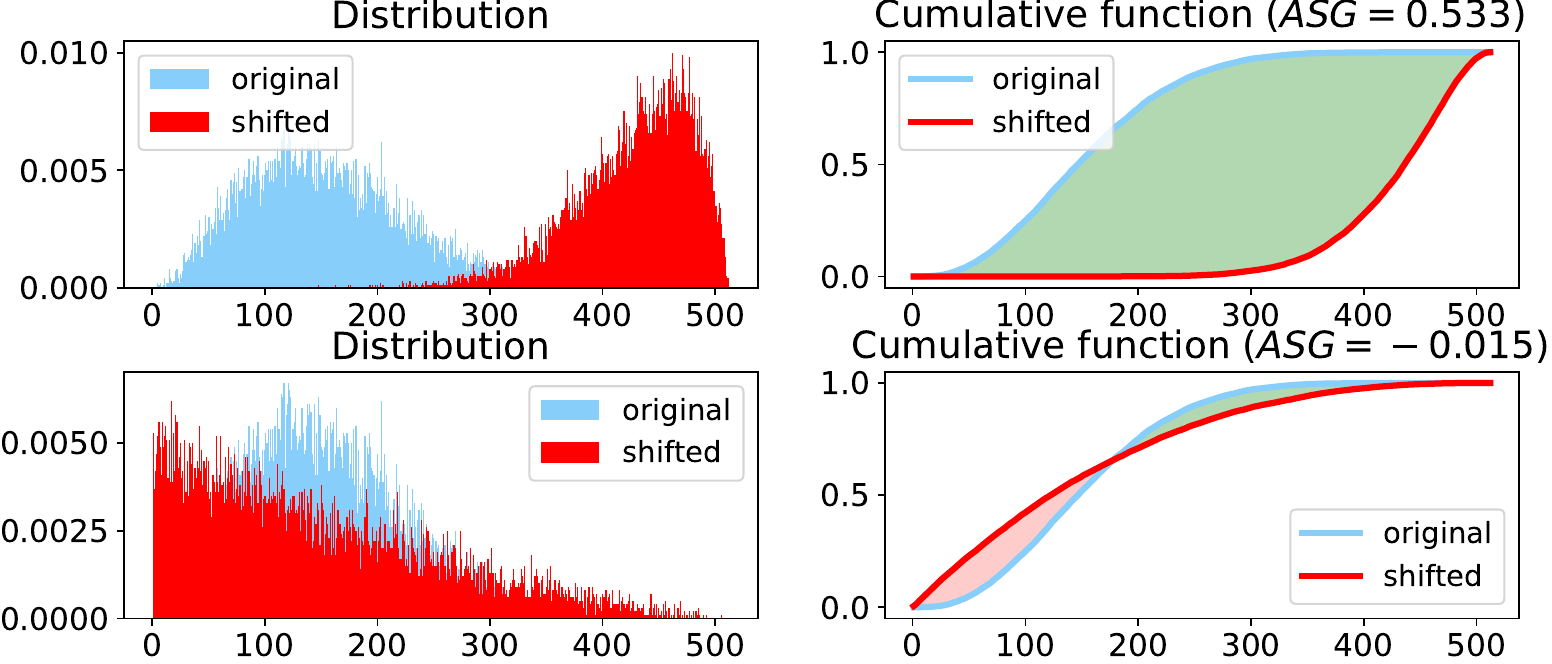}
    \caption{Examples of \ASG. The difference between the cumulative functions of the original and shifted distributions is positive in the green area and negative in the red area.}
    \label{fig:metric_intuition}
    \vspace{-10pt}
\end{figure}
A straightforward metric for attack efficacy is Wasserstein distance, which measures the distance between the estimated distribution with the poisoning attack and the clean one. 
However, it cannot reflect the direction in which the attacked distribution is skewed.
Any local skews on the distribution, whether towards the right or left, can increase the Wasserstein distance.
To address this problem, we adapt the definition of $L_1$ Wasserstein distance but replace the $L_1$ norm with differences with signs to form our new metric. 
Denote the original and the estimated distributions as $X = [\Phi(B_i)]_{i=1}^{m_o}$ and $\hat{X}_a = [\Phi(\hat{B}_i)]_{i=1}^{m_o}$ respectively. \ASG can be formally defined as $\ASG(\hat{X}_a) = \sum_{v=1}^{m_o} P(X, v) - P(\hat{X}_a, v)$,
where $P(X, v) = \sum_{i=1}^{v} \Phi(B_i)$ is the cumulative function over distribution $X$. The intuition behind the metric is shown in Figure~\ref{fig:metric_intuition}. As the distribution is shifted to the right end, the value of the cumulative function on the left side of the domain gets smaller while becoming larger on the right side. The farther the distribution is shifted to the right end, the larger $\sum_{v=1}^{m_o} P(X, v)$ is than $\sum_{v=1}^{m_o} P(\hat{X}_a, v)$. Thus, larger \ASG indicates a better shift effect, thus better attack efficacy. Besides, negative \ASG indicates the overall distribution is shifted to the left side.

\noindent\textbf{Shift Gain Ratio (\SGR).}
The \ASG is a straightforward metric that aligns with the attacker's goal and can measure the ultimate effectiveness of attacks on the numerical LDP protocols.
However, it is implicitly influenced by many factors, including how many users are corrupted by the attacker and what the true data distribution is.
Besides, as we look into the attacks that are more effective than the baseline attacks, \ASG fails to demonstrate the advantage of the proposed attacks over the baseline.
Thus, we propose a new metric, \SGR, to measure the effectiveness of the proposed attacks compared to the baselines and thus enable the analysis of the robustness of LDP protocols.

We first consider the ratio of the shift gain of our attack to that of the baseline attack, i.e., $\frac{\ASG(X, \hat{X}_{a})}{\ASG(X, \hat{X}_{a}^{base})}$, where $\hat{X}_{a}^{base}$ is the estimated distribution after baseline attack. Since the baseline is universal regardless of the underlying protocol, comparing the attack effectiveness using SGR is viable across different protocols, thus enabling the assessment of protocol robustness. However, directly using this ratio as a metric may lead to an unstable measurement of additional variance because the denominator includes $\hat{X}_{a}^{base}$, a random variable with LDP noise.
We replace the term $\hat{X}_{a}^{base}$ with the skewed distribution in the input domain $X_{a}^{base}$.
Note that this is a reasonable substitution because the estimated distribution in well-designed LDP protocols should be close to the distribution in the input domain. Formally, 
\begin{align}
    \SGR(\hat{X}_a) = \frac{\ASG(X, \hat{X}_{a})}{\ASG(X, X_{a}^{base})}
    \label{eq:proposed_attack_metric}
\end{align}

By Equation \eqref{eq:proposed_attack_metric}, a smaller value of the metric means lower attack efficacy. Specifically, the metric equal to 1 indicates the attack performs similarly to the baseline. When the metric is less than 1, the attack underperforms the baseline. Moreover, when the ratio is negative, the attack efficacy tends to move in a different direction than the baseline, e.g., the baseline shifts the distribution to the right end, but the measured attack is not able to achieve this goal.

With a predetermined dataset and a fixed number of corrupted users, \SGR is proportional to \ASG. Thus, \SGR essentially measures the attack efficacy improvement per fake user in the proposed attack versus the baseline, that is, how many fake users are needed in the baseline to achieve the same effect of one fake user in our attack. When $\SGR = \frac{1}{\beta}$, $\beta$ fake users in our attack equals $\beta \times \frac{1}{\beta} = 100\%$ malicious users in baseline. In this case, the frequency of the right-most bin is 1. Therefore, $\frac{1}{\beta}$ is the upper bound of \SGR.
Besides, the normalization with the baseline attack gain brings two benefits. The first is that \SGR is more consistent for the protocols across different datasets.
The other is that \SGR can show how the attack efficacy is affected by different portions of corrupted users, while the \ASG will monotonically increase with the corrupted portion.

\subsection{Attack on CFO with binning and consistency}
Optimal attacks for shifting distribution with different CFOs depend on the real data distribution, which may not be known to the attacker in general. However, the main idea of shifting the distribution to the right side is to increase the frequencies of large values. For robustness analysis, we consider a heuristic attack that increases the frequency of the largest value $\Phi(B_{m_o})$ as much as possible. This attack can ensure the distribution is shifted in the direction as attacker expected for all distribution and empirically achieves significant attack efficacy. We describe $\Phi(B_{m_o})$ with $n_f$ malicious users and the creation of fake reports to each protocol.

\subsubsection{Attacking GRR} The aggregated result $\Phi(B_{m_o})$ and corresponding fake reports are as follows.
\begin{itemize}[leftmargin=*]
    \item \textbf{Aggregation $\Phi(B_{m_o})$:} The estimated frequency $\Phi(B_{m_o}) = \frac{\sum_{j=1}^{n_g} \mathbb{I}_{[m_o]}(y^{(j)}) + \sum_{j=1}^{n_f} \mathbb{I}_{[m_o]}(\hat{y}^{(j)}) - \frac{n}{e^\epsilon + m_o - 1}}{\frac{n(e^\epsilon - 1)}{e^\epsilon + m_o - 1}}$.
    \item \textbf{Fake report craft:} The most effective way to promote $\Phi(B_{m_o})$ is to set all fake values $\hat{y}^{(j)} = m_o$ such that the item $\sum_{j=1}^{n_f} \mathbb{I}_{[m_o]}(\hat{y}^{(j)}) = n_f$.
\end{itemize}

\subsubsection{Attacking OUE} The aggregated result $\Phi(B_{m_o})$ and corresponding fake reports are as follows.
\begin{itemize}[leftmargin=*]
    \item \textbf{Aggregation $\Phi(B_{m_o})$:} The estimated frequency $\Phi(B_{m_o}) = \frac{\sum_{j=1}^{n_g} \hat{\bm{v}}^{(j)}[m_o] + \sum_{j=1}^{n_f} \hat{\bm{y}}^{(j)}[m_o] - \frac{n}{e^\epsilon + 1}}{\frac{1}{2} - \frac{1}{e^\epsilon + 1}}$.

    \item \textbf{Fake report craft:} Since the compromised users in OUE can craft the fake report vector $\hat{\bm{y}}$, the most effective way to promote $\Phi(B_{m_o})$ is to craft $\hat{\bm{y}} = [0, 0, .., 1]$ such that only the frequency of the $m_o$-th bin is contributed.
\end{itemize}

\subsubsection{Attacking HST}
There are two different implementations of HST depending on whether users or the server samples the public vector $s^{(i)} \in \{\pm 1\}^{m_o}$. Without the attack, these two implementations give the same privacy and utility guarantee.
However, our experiments show that when the server generates the public vector for each user, the attack can be largely constrained on his attack efficacy.
Thus, we elaborate on our two implementations as follows.

\mysubpara{\textbf{HST-User}.}
The aggregated result $\Phi(B_{m_o})$ and corresponding fake reports are as follows.
\begin{itemize}[leftmargin=*]
    \item \textbf{Aggregation $\Phi(B_{m_o})$:} When the users can select public vectors by themselves, they can pick binary vectors that promote the frequency of the largest value in the domain. Thus, the estimated frequency of the $m_o$-th bin is $\Phi(B_{m_o}) = \frac{1}{n} (\sum_{j=1}^{n_g} \hat{s}^{(j)}_{x_b} \times \bm{s}^{(j)}[m_o] + \sum_{j=1}^{n_f} \hat{y}^{(j)} \times \hat{\bm{s}}^{(j)}[m_o])$, where $\hat{\bm{s}}^{(j)}$ is the binary vector picked by malicious user $j$.

    \item \textbf{Fake report craft:} The idea of promoting $\Phi(B_{m_o})$ is to craft fake value $\hat{y}^{(j)}$ and carefully select public binary vector $\hat{\bm{s}}^{(j)}$ so that the aggregated result $\hat{y}^{(j)} \times \hat{\bm{s}}^{(j)}[m_o]$ only contributes to $\Phi(B_{m_o})$. Specifically, the attacker sets the public binary vector as $\bm{s}^{(j)} = [-1, -1, ..., 1]$ where only the right-most is 1, and then crafts fake value $\hat{y}^{(j)}$ as $\frac{e^\epsilon + 1}{e^\epsilon - 1}$ to maximally increase the frequency of the right-most bin.
\end{itemize}

\mysubpara{\textbf{HST-Server}.}
The server sets the public binary vector for each user uniformly at random at the beginning of the protocol. Thus, the aggregated result $\Phi(B_{m_o})$ and corresponding fake reports for HST-Server are different from HST-User.
\begin{itemize}[leftmargin=*]
    \item \textbf{Aggregation $\Phi(B_{m_o})$:} Each user in HST-Server must use the assigned binary vector. Thus, the attacker is constrained to use $\bm{s}^{(j)}$ for its corrupted user $j$. Consequently, the estimated frequency of the $m_o$-th bin is $\Phi(B_{m_o}) = \frac{1}{n} (\sum_{j=1}^{n_g} \hat{s}^{(j)}_{x_b} \times \bm{s}^{(j)}[m_o] + \sum_{j=1}^{n_f} \hat{y}^{(j)} \times \bm{s}^{(j)}[m_o])$.
    
    \vspace{2pt}
    \item \textbf{Fake report craft:} Attacker can only manipulate the fake value $\hat{y}^{(j)}$ as the same value as $\frac{e^\epsilon + 1}{e^\epsilon - 1} \times \bm{s}^{(j)}[m_o]$ and in this way the aggregated result $\hat{y}^{(j)} \times \bm{s}^{(j)}[m_o]$ can promote the largest value in the domain. However, since the vector follows uniform distribution and each report supports about half of the bins, the attacker cannot increase the frequency of the largest bin only.
\end{itemize}

\subsubsection{Attacking OLH}
There are also two implementations for OLH depending on whether users or the server chooses the hash function, which become equivalent in terms of privacy and utility guarantees without the attack.
Given a hash function $H$, the attacker can search the hash mappings and learn the inverse of the hash function $H^{-1}$, mapping a hash output in $[g]$ to a subset of $[m_o]$.
Our experiments also show that when the server decides the hash functions for each user, the attacker can be largely constrained on its manipulation efficacy. The two implementations are as follows.

\mysubpara{\textbf{OLH-User}.}
The aggregated result $\Phi(B_{m_o})$ and corresponding fake reports are below.
\begin{itemize}[leftmargin=*]
    \item \textbf{Aggregation $\Phi(B_{m_o})$:} When the users randomly choose the hash function by themselves, the attacker can choose the hash functions from a set of all hash functions $\mathbf{H}$ for the compromised users. Thus, the aggregated result $\Phi(B_{m_o}) = \frac{C_{g}(B_{m_o}) + C_{f}(B_{m_o}) - \frac{n}{g}}{n(\frac{e^\epsilon}{e^\epsilon + g - 1} - \frac{1}{g})}$, where $C_{g}(B_{m_o}) = |\{ j | H^{(j)}(m_o) = \hat{x}_h^{(j)}, j \in \mathrm{genuine\,users} \}|$, $C_{f}(B_{m_o}) = |\{ j | \hat{H}^{(j)}(m_o) = \hat{y}^{(j)}, j \in \mathrm{malicious\,users} \}|$ and $\langle \hat{H}^{(j)}, \hat{y}^{(j)} \rangle$ is fake report of malicious user $j$.

    \item \textbf{Fake report craft:} Ideally, the attacker hopes to craft a fake report $\langle \hat{H}^{(j)}, \hat{y}^{(j)} \rangle$ where the hash function only maps the largest value into $\hat{y}^{(j)}$. However, they may not be able to find such a function within a given amount of time in the set $\mathbf{H}$. To address this problem, we randomly sample 1,000 hash functions as the candidate set in which we find each $\hat{H}^{(j)}$. The fake report should satisfy two conditions: 1) $\hat{H}^{(j)}(m_o) = \hat{y}^{(j)}$ and 2) all values mapped into $\hat{y}^{(j)}$ should concentrate on the higher end as much as possible, i.e., the mean of these values should be largest among all hash functions in $\mathbf{H}$. The first condition guarantees that the fake report promotes the frequency of the largest value and the second ensures that the fake report does not contribute to frequencies of small values and weakens attack efficacy.
\end{itemize}

\mysubpara{\textbf{OLH-Server}.}
OLH-Server has different aggregated result $\Phi(B_{m_o})$ and corresponding fake reports since the server picks a hash function for each user at the beginning of the protocol.
\begin{itemize}[leftmargin=*]
    \item \textbf{Aggregation $\Phi(B_{m_o})$:} Since all users are constrained to use the assigned hash functions, the aggregated result $\Phi(B_{m_o}) = \frac{C_{g}(B_{m_o}) + C_{f}(B_{m_o}) - \frac{n}{g}}{n(\frac{e^\epsilon}{e^\epsilon + g - 1} - \frac{1}{g})}$, where $C_{g}(B_{m_o})$ is the same as that in OLH-User but $C_{f}(B_{m_o}) = |\{ j | H^{(j)}(m_o) = \hat{y}^{(j)}, j \in \mathrm{malicious\,users} \}|$ and $\langle H^{(j)}, \hat{y}^{(j)} \rangle$ is a fake report of malicious user $j$.

    \item \textbf{Fake report craft:} The attacker is constrained to use the assigned $H^{(j)}$ for its corrupted user. Thus, the attacker only needs to find $\hat{y}^{(j)}$ such that $H^{(j)}(m_o) = \hat{y}^{(j)}$.
\end{itemize}

\subsection{Attack on SW}
The convergence point of EMS algorithm depends on the privacy budget $\epsilon$ and the underlying data, which are unknown to the attacker \cite{li2020estimating}. Consequently, the recovered distribution is unpredictable. It is also challenging to find the optimal attack for SW. Similar to the CFO-based methods, the solution to shifting the distribution under SW is still to promote the frequencies of large values in the domain while reducing other values' frequencies accordingly.

We consider an intuitive way to achieve the attack in practice. We inject the fake values into the bins near the right end of the output domain $[-b, 1+b]$ to promote the frequencies of large values. Specifically, we uniformly at random inject the fake values into 1) the right-most bin, 2) range $[1+\frac{2b}{3}, 1+b]$, 3) range $[1, 1+b]$ and 4) range $[1-b, 1+b]$.



\vspace{-.05in}
\begin{tcolorbox}[enhanced, drop shadow southwest, left=1pt, right=1pt, top=0pt, bottom=0pt]
Regarding RQ-1, we find no one-size-fits-all attack strategy for various LDP protocols. However, it is possible to use unified metrics to measure LDP robustness regardless of the design differences of the underlying protocols. We propose two new metrics, i.e., \ASG and \SGR, to quantify the attack strength for fair robustness evaluation, which may also benefit future research.  
\end{tcolorbox}
\vspace{-.05in}
\section{Robustness Evaluation} \label{sec:attack_experiments}

\begin{figure*}[!htbp]
    \centering
    \includegraphics[scale = 0.35]{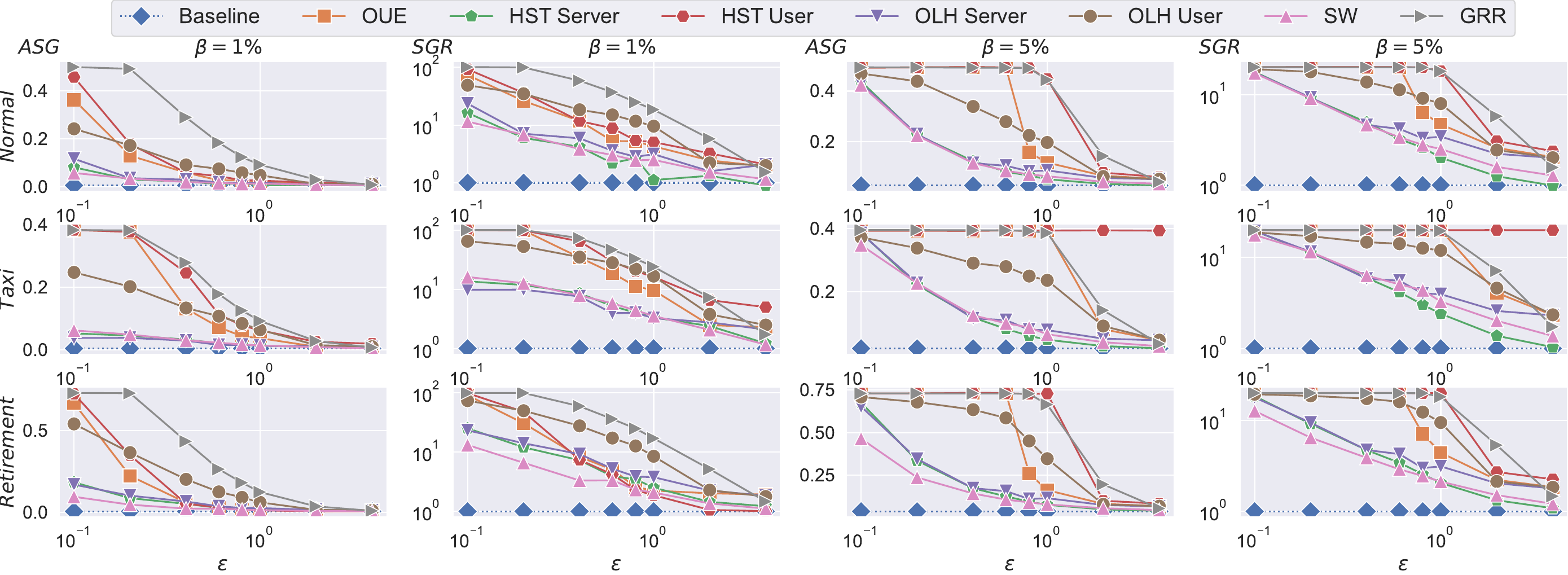}
    \vspace{-7pt}
    \caption{Attack results with varying $\epsilon$ from $0.1$ to $4$. Each row corresponds to one dataset. The left two columns show \ASG and \SGR with $\beta = 1\%$ and the right two columns depict \ASG and \SGR with $\beta = 5\%$.}
    \label{fig:main_result_beta_eps}
    \vspace{-10pt}
\end{figure*}

\subsection{Experiment Setup}
\noindent\textbf{Datasets.}
We use one synthetic dataset and two real-world datasets to conduct our experiments. The dataset information is summarized below and the distribution of the datasets is shown in Figure~\ref{fig:data_dist} in Appendix~\ref{app:supplemental_experiment_results}.
\begin{enumerate}[leftmargin=*]
    \item \textit{Synthetic Gaussian dataset.} We draw $10^5$ samples from normal distribution $\mathcal{N}(0, 10)$ to generate the synthetic dataset.
    
    \item \textit{Taxi} \cite{taxi}. This dataset was published by the New York Taxi Commission in 2018. It contains 2,189,968 samples of pickup time in a day (in seconds).
    
    
    \item \textit{Retirement} \cite{retirement}. This dataset contains data about San Francisco employee retirement plans, ranging from $-28,700$ to 101,000. We extract non-negative values smaller than 60,000 for evaluation. After pre-processing, there are 178,012 samples in our experiments.
\end{enumerate}
It is worth noting that all mechanisms map the data into $[0, 1]$. Therefore, we linearly map the data into the corresponding range for perturbation.

\noindent\textbf{Parameter Settings.}
The attack efficacy depends on privacy budget $\epsilon$ and the fraction of compromised users $\beta$. We vary $\epsilon$ from 0.1 to 4 since they are common values in LDP. We also set $\beta$ from $1\%$ to $7.5\%$ to study the impact of each parameter on robustness. We also set the number of bins for CFO-based protocols and SW as $m_o = 32$ and $m_s = 512$ respectively, because this setting achieves the best performance in most cases \cite{li2020estimating}. For each dataset and each attack, we repeat the experiment 100 times and show the average result.

\subsection{Results}

\subsubsection{Robustness comparison}
In this subsection, we study the impact of varying parameters on the robustness of LDP protocols. 

\noindent\textbf{Impact of $\epsilon$.}
We first empirically study the impact of $\epsilon$ on the robustness of LDP protocols. We attack the SW mechanism by selecting four different ranges and choosing the one that yields the highest \SGR with varying $\epsilon$ for a fair comparison with other protocols. We defer the discussion on the attack performance on SW with different ranges to Section~\ref{sec:attack_efficacy_SW}. Figure~\ref{fig:main_result_beta_eps} shows the attack efficacy with three datasets. We have the following key observations.

\begin{figure*}[!htbp]
    \centering
    \includegraphics[scale = 0.35]{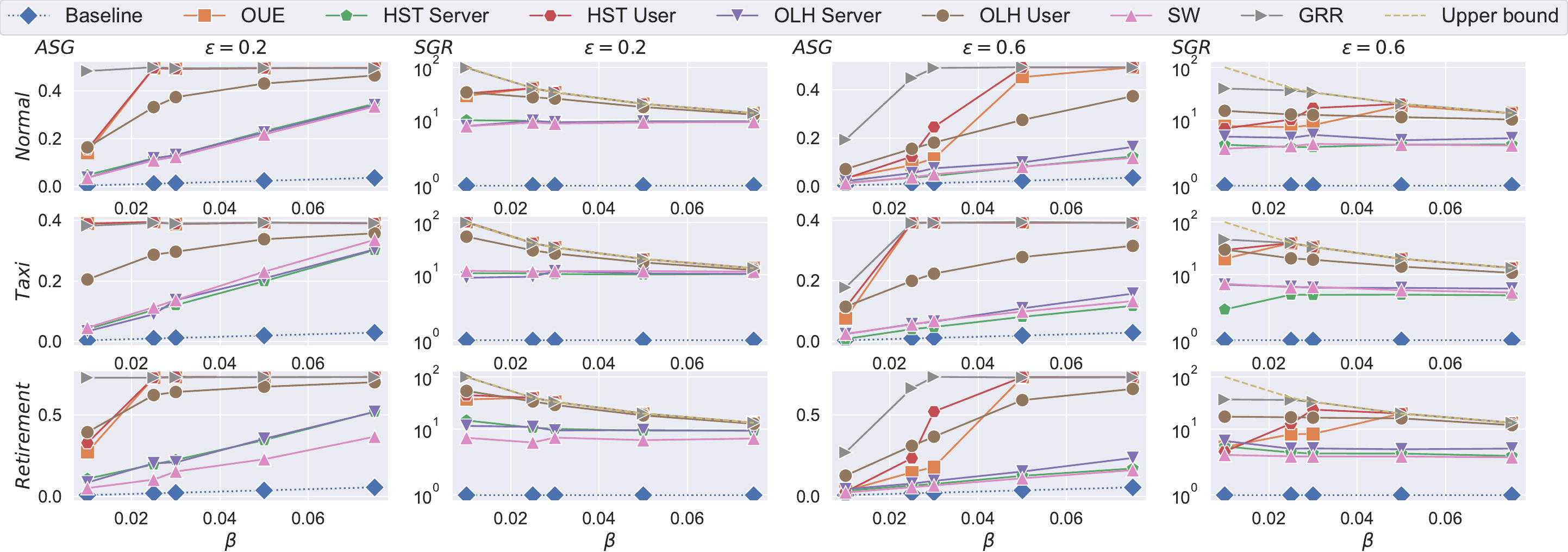}
    \vspace{-7pt}
    \caption{Attack results with varying $\beta$ from $1\%$ to $7.5\%$. Each row corresponds to one dataset. The left two columns show \ASG and \SGR with $\epsilon = 0.2$ and the right two columns depict \ASG and \SGR with $\epsilon = 0.6$.}
    \label{fig:main_result_eps_beta}
    \vspace{-10pt}
\end{figure*}

\begin{itemize}[leftmargin=*]
    \item In most cases, SW and \textit{Server} setting of HST and OLH have the smallest \ASG and \SGR and thus have better natural resilience against the attack. 
    \begin{itemize}[leftmargin=*]
        \item For SW, the reason is that the fake values promote the frequencies of bins near the right end of the domain. The EMS can average the frequencies of polluted bins with other bins and thus reduce their frequencies, making the attack efficacy drop.

        \item For \textit{Server} setting, the hash functions in OLH-Server and public binary vector in HST-Server are assigned by the server and the malicious user cannot ensure that the fake value only supports higher-end bins.
    \end{itemize}
    

    \item GRR, OUE and \textit{User} setting of OLH and HST provide the weakest robustness. This is because in HST-User, GRR and OUE, malicious users can set fake reports to only contribute to the frequency of the right-most bin, and fake reports for OLH-User can also contribute to bins at the higher end of the domain.


    \item For small $\beta$ (e.g., $\beta = 1\%$), OLH-User sometimes is more robust than GRR, OUE and HST-User but its robustness becomes weaker for large $\beta$. This is because the attack only focuses on a single bin (i.e., the right-most bin in the domain) for GRR, OUE and HST-User, while bogus data for OLH-User would support a range of values on the right end. Thus, the attack performs better on OLH-User with a small number of malicious users.

    \item Although GRR, OUE and HST-User can only support the right-most bin, OUE shows better robustness, especially for large $\epsilon$ given enough compromised users (e.g., $\beta = 5\%$). This is because the denominator in $\Phi_{\mathrm{OUE}}(\cdot)$ increases as $\epsilon$ grows, leading to small frequencies of higher-end bins and weak attack efficacy.
    
    \item \ASG (\SGR) reduces as $\epsilon$ grows and the attack efficacy is close to the baseline given a larger $\epsilon$ (e.g., $\epsilon \geq 1$). This is because a larger $\epsilon$ leads to smaller noise, and thus the results recovered by aggregation are closer to the input values and \ASG (\SGR) is closer to the baseline.

    \item GRR shows the worst robustness for small $\beta$ (e.g., $\beta = 1\%$) and the gap between GRR and other protocols is large. This is because the denominator of $\Phi_{\mathrm{GRR}}(\cdot)$ is much smaller than others, promoting the frequency of the right-most bin in attack. However, the gap is small for large $\beta = 5\%$ since fake values in other protocols also influence frequencies of other bins, which helps shift the distribution with consistency and fill the gap with GRR.

    \item The CFO-based methods on \textit{Taxi} performs the worst against the attack, and the \ASG (\SGR) is the largest over a wide range of $\epsilon$. This is because the distribution of the dataset is ``flat'', i.e., there is no significantly high probability density region. Therefore, even under a relatively large $\epsilon$, the frequencies of all bins except for the polluted bins could be negative under the attack. After the Norm-Sub, only the polluted bins have positive frequencies.
\end{itemize}

\noindent\textbf{Impact of $\beta$.}
We then empirically study the impact of $\beta$.
We plot the upper bound $\frac{1}{\beta}$ of \SGR and show the attack effectiveness as a function of $\beta$ in Figure~\ref{fig:main_result_eps_beta}. We have the following observations.

\begin{itemize}[leftmargin=*]
    \item In general, attack effectiveness increases as $\beta$ grows and \SGR asymptotically reaches the upper bound. However, SW, HST-Server and OLH-Server still have the smallest \ASG (\SGR) given various $\beta$ in most cases due to the EMS and malicious users in \textit{Server} setting cannot ensure that fake values only support higher-end bins.

    \item Due to the same reason, GRR, OUE, OLH-User, and HST-User provide the weakest robustness since malicious users can set fake reports to only contribute to the frequency of the higher-end bins.

    \item We also observe that \SGR of GRR, OUE and HST-User reach the upper bound when $\beta$ is sufficiently large, especially with \textit{Taxi}, while OLH-User does not with a large $\beta$. This is because the fake values in GRR, OUE and HST-User only support the right-most bin. But in OLH-User, its fake values do not support the maximal $x_b$ only. Thus, the estimate cannot reach the maximum of the domain.

    \item The results show that the proposed metric \SGR is stable with varying $\beta$. This is because \SGR measures the attack improvement per compromised user compared with the baseline, irrelevant to the number of malicious users.




\end{itemize}

\subsubsection{Robustness and $g$ Trade-off} \label{sec:security_g_relationship}
We are still in the early stage of understanding the threat landscape. It is only known that LDP security has a strong connection with user privacy preference in prior work \cite{li2022fine, cao2021data}, i.e., tuning $\epsilon$ could enhance security. However, it is not the only factor deciding the protocol's robustness. For the local-hashing-based LDP, we found the hash domain size $g$ also has a profound influence.
\begin{theorem}
    For local-hashing-based CFOs with fixed $\epsilon$, the expected \ASG becomes lower (higher) when the hash domain $g$ is smaller (larger) before post-processing.
\label{the:security_g_relationship}
\end{theorem}
\begin{proof}
    See Appendix~\ref{app:proof_security_g_relationship}
\end{proof}
The intuition is that in the aggregation, each report contributes to the frequencies of $\frac{m_o}{g}$ bins on average, which becomes smaller as $g$ grows. Consequently, it is easier for the attacker to manipulate the higher-end bins, leading to enhanced attack performance. 

Theorem~\ref{the:security_g_relationship} proves the relationship between attack efficacy and $g$ without the post-processing that has no closed-form solution. However, our experiment further shows that the relationship remains in Figure~\ref{fig:ASG_and_g}. We have similar observations across different $\epsilon$ and $\beta$ and only show the result with $\epsilon = 0.2$ and $\beta = 5\%$. We also observe that attack efficacy under \textit{User} setting and \textit{Server} setting tend to be close as $g$ grows. This is because each report can support fewer bins under larger $g$. When $g$ is excessively large, fake data only contributes to higher-end bins, leading to weaker robustness (higher \ASG) even under \textit{Server} setting.

\begin{figure}[H]
    \centering
    \includegraphics[scale = 0.34]{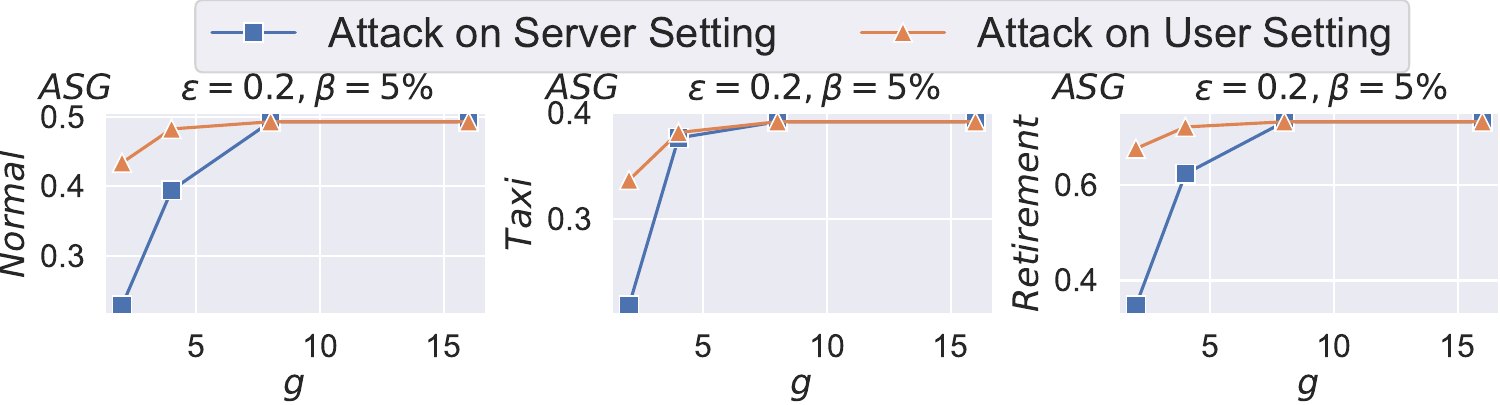}
    \vspace{-5pt}
    \caption{Relationship between attack efficacy and hash domain size $g$.}
    \label{fig:ASG_and_g}
\end{figure}

\begin{figure*}[!t]
    \centering
    \includegraphics[scale = 0.35]{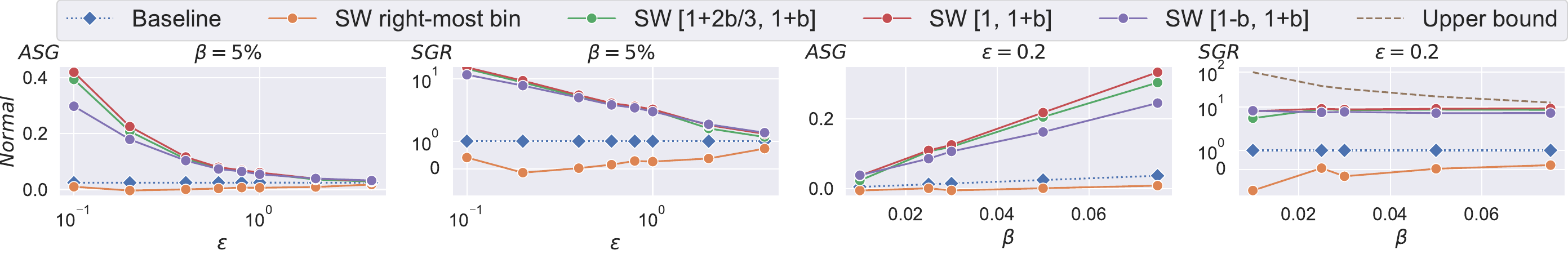}
    \vspace{-7pt}
    \caption{Attack results on SW mechanism on dataset $\mathcal{N}(0, 10)$, varying $\epsilon$ from $0.1$ to $4$. The fake values are injected into 1) the right-most bin, 2) range $[1+\frac{2b}{3}, 1+b]$, 3) range $[1, 1+b]$ and 4) range $[1-b, 1+b]$.}
    \label{fig:SW_bin_eps_beta}
    \vspace{-10pt}
\end{figure*}

\subsubsection{A Closer Look at Attack on SW} \label{sec:attack_efficacy_SW}
We empirically study the efficacy of all four attack strategies on the SW mechanism. We have similar observations across all three datasets and only show the result with the dataset $\mathcal{N}(0, 10)$  in Figure~\ref{fig:SW_bin_eps_beta} due to space limitations. The rest are shown in Figure~\ref{fig:SW_bin_eps_beta_full_results} in Appendix~\ref{app:supplemental_experiment_results}. Overall, except for injecting fake values into the right-most bin, other attack types perform similarly, and all of them can effectively increase the \ASG/\SGR of the baseline and achieve satisfactory attack performance. In particular, the \SGR is greater than 10 with $\epsilon = 0.1$. This is because these types of attacks increase the probability density region of the right end of the domain, thus skewing the distribution substantially. We choose the attack range $[1, 1+b]$ as the best attack in this paper since it is the empirically strongest one.

We also observe that injecting fake values into the right-most bin performs the worst in our experiments. Its \SGR is less than that of the baseline in most cases and even negative sometimes. Moreover, increasing $\beta$ does not lead to an obvious improvement in attack effectiveness because the fake values are only in one bin, and the smoothing step in EMS can significantly reduce the anomaly high frequency of this bin even with a large $\beta$. In addition, since injecting fake values into only one bin is largely impacted by EMS whose output is unpredictable, the relationship between attack efficacy and $\epsilon$ is also uncertain.

The third observation is that the attack efficacy is close to the baseline with a large $\epsilon$. This is because a larger $\epsilon$ leads to smaller noise. Thus the results recovered by aggregation are closer to the input values and the \SGR approaches that of the baseline.



\vspace{-.05in}
\begin{tcolorbox}[enhanced,drop shadow southwest, left=1pt, right=1pt, top=0pt, bottom=0pt]
For RQ-2, our research reveals that different LDP protocols indeed exhibit distinct robustness in the presence of the attack. SW and CFO-based protocols in \textit{Server} setting are more resilient due to the additional smoothing step and constrained resources. Our investigation generates new knowledge about the relationship between LDP implementation and security, i.e., the hash domain size $g$ in OLH will affect the attack resilience. The result highlights the need to study attack-resilient post-processing and fine-tune related parameters for real-world LDP implementation.
\end{tcolorbox}
\vspace{-.05in}

\section{Zero-shot Attack Detection}

We illustrate our detection method in this section. We call it \textit{zero-shot attack detection} because it does not require any ground truth data, similar to zero-shot learning~\cite{zeroshot}.

\vspace{-10pt}
\subsection{Detection Details}
\vspace{-3pt}


In a scenario without poisoning attacks, the recovered distribution should be close to the original distribution.
However, in the presence of poisoning attacks, attackers submit fake values directly to the data collector instead of executing LDP random perturbation. This results in the distribution of poisoned reports being a very low likelihood outcome of any valid input distribution. In other words, it is unlikely to observe a perturbed distribution similar to that of poisoned reports given any valid input distribution.

Based on the above intuition, we present a novel detection mechanism by reconstructing a comparison benchmark. The benchmark is created using the distribution information returned by the LDP mechanism. Without the attack, recovered data distribution should be close to the original. We denote the synthesis function as $\mathcal{S}(\hat{X}, t)$, which takes noisy report $\hat{X}$ and an integer $t$ as input, then estimates a distribution from $\hat{X}$ and outputs $t$ samples drawn from recovered distribution. The noisy report $\hat{X}$ resulting from applying LDP perturbation to the original distribution, and $\hat{X}_2$ resulting from applying the same LDP randomizer to $\mathcal{S}(\hat{X}, t)$ should have similar distribution. When applying $\mathcal{S}$ to the noisy report $\hat{X}_2$, one should again recover the similar distribution, and thus $\hat{X}_3$ generated by applying again the LDP randomizer to $\mathcal{S}(\hat{X}_2, t)$ should also result in distributions similar to $\hat{X}_2$. Hence, we can detect the manipulation by measuring the two groups of distances, one is between $\hat{X}$ and $\hat{X}_2$ denoted as $g_{det}$ and the other is between $\hat{X}_2$ and $\hat{X}_3$ denoted as $g_{ben}$. Since $\hat{X}_2$ and $\hat{X}_3$ are generated without manipulation and their distribution should always be close, we consider $g_{ben}$ as the benchmark. Under the attack, the polluted $\hat{X}$ is low-likelihood outcome of any valid input, and applying the same LDP randomizer to such $\mathcal{S}(\hat{X}, t)$ cannot output a similar distribution to that of $\hat{X}$, leading to that $g_{det}$ statistically different from $g_{ben}$.
Figure~\ref{fig:detection_workflow} shows the detection workflow. The detection is described in Algorithm~\ref{alg:detection}.

\begin{figure}[!t]
    \centering
    \includegraphics[scale = 0.23]{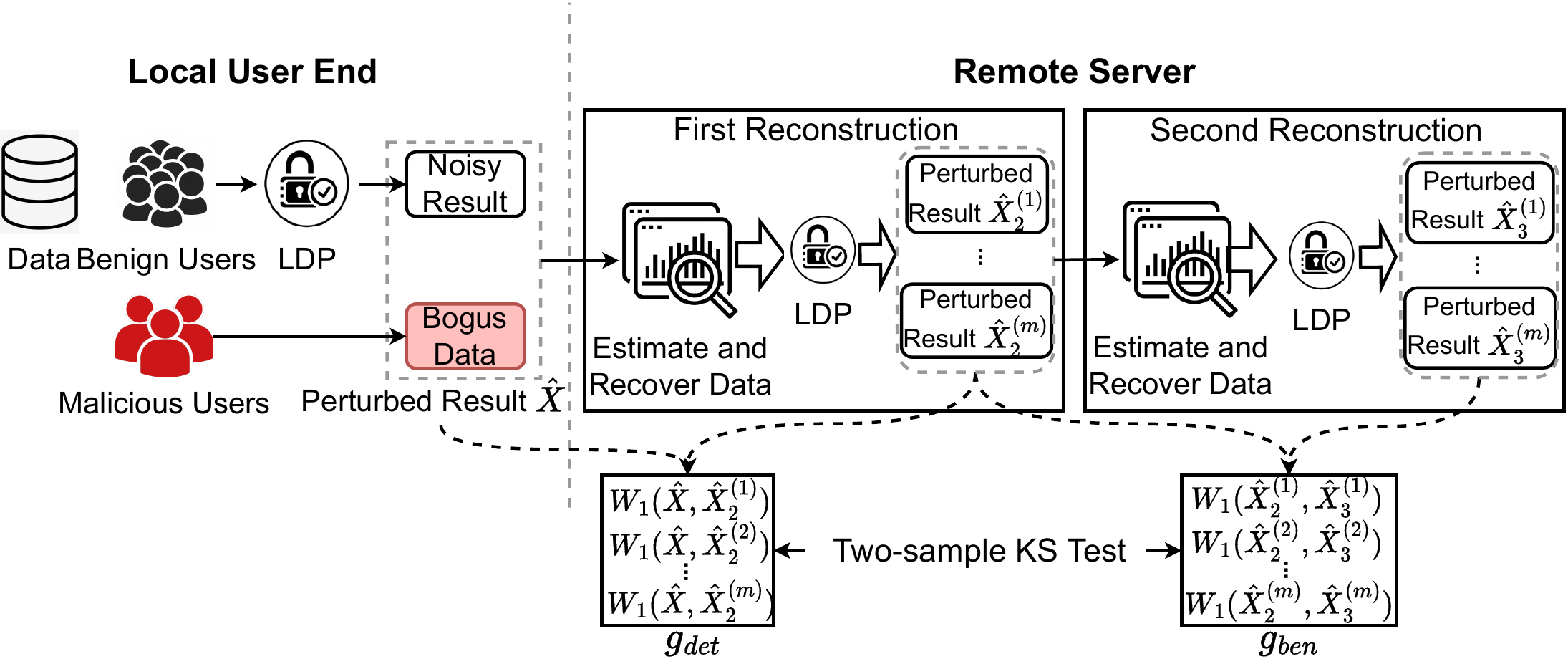}
    \caption{Zero-shot Detection Workflow.}
    \vspace{-25pt}
    \label{fig:detection_workflow}
\end{figure}


Formally, we assume there are $n$ users and denote the reported perturbed result as $\hat{X}$. After receiving $\hat{X}$, the server can estimate the data distribution and draw $n$ samples from the distribution. Following the LDP protocol, the server perturbs the sampled data $m$ times and gets a group of reconstructed noisy results $\hat{\bm{X}}_2 = [\hat{X}_2^{(i)}]_{i=1}^{m}$. We generate $m$ noisy results because the perturbation is random and multiple results can decrease statistical uncertainty. Given the noisy result $\hat{\bm{X}}_2$, the server repeats this reconstruction process to derive the noisy result $\hat{\bm{X}}_3 = [\hat{X}_3^{(i)}]_{i=1}^{m}$. Then we calculate the pairwise distribution distance between  $\hat{X}_2^{(i)}$ and $\hat{X}_3^{(i)}$ and form a group of distances $[W_1(\hat{X}_2^{(i)}, \hat{X}_3^{(i)})]_{i=1}^{m}$. We consider the Wasserstein distance (denoted as $W_1(\cdot)$) since it measures the overall difference between two distributions. Since we honestly follow the LDP protocol to generate $\hat{\bm{X}}_2$ and $\hat{\bm{X}}_3$, the distance $W_1(\hat{X}_2^{(i)}, \hat{X}_3^{(i)})$ should be small and is used as the benchmark $g_{ben}$ to compare with the distance group $[W_1(\hat{X}, \hat{X}_2^{(i)})]_{i=1}^{m}$ denoted as $g_{det}$. 

\begin{algorithm}[!htbp]
	\caption{Zero-shot Attack Detection} \label{alg:detection}
		\textbf{Input:} Noisy results $\hat{X}$, significance level $\alpha$, $\epsilon$, $n$. \\
		\textbf{Output:} Whether the reported data is polluted.
	\begin{algorithmic}[1]
		\State Initialize $\hat{\bm{X}}_2 = \hat{\bm{X}}_3 = \emptyset$.
		\State Get synthetic data $X = \mathcal{S}(\hat{\bm{X}}, n)$
		\For{$i \in [m]$}   \Comment{Reconstruct noisy results}
		\State $\hat{X}_2^{(i)} \leftarrow$ Perturb $X$ by LDP randomizer with $\epsilon$.
		\State Get synthetic data $X_2 = \mathcal{S}(\hat{X}_2^{(i)}, n)$
		\State $\hat{X}_3^{(i)} \leftarrow$ Perturb $X_2$ by LDP randomizer with $\epsilon$.
		\State $\hat{\bm{X}}_2.{append}[\hat{X}_2^{(i)}], \hat{\bm{X}}_3.{append}[\hat{X}_3^{(i)}]$
		\EndFor
		\State Benchmark distance group $g_{ben} \leftarrow [W_1(\hat{X}_2^{(i)}, \hat{X}_3^{(i)})]_{i=1}^{m}$
		\State Distance group to be detected $g_{det} \leftarrow [W_1(\hat{X}, \hat{X}_2^{(i)})]_{i=1}^{m}$
		\State $p$-value $\leftarrow \mathrm{Two\_Sample\_KS\_Test}(g_{det}, g_{ben})$.
		\If {$p$-value $< \alpha$}
		\State \Return The reported result is polluted.
		\Else
		\State \Return The reported result is unpolluted.
		\EndIf
	\end{algorithmic}
\end{algorithm}

We statistically analyze the similarity between $g_{det}$ and $g_{ben}$. Intuitively, $g_{det}$ and $g_{ben}$ should follow the same or similar distribution if the reported data are clean. Therefore, we formulate the detection as a hypothesis test problem. The \textit{null hypothesis} is that the reported data is not polluted, i.e., $g_{det}$ and $g_{ben}$ come from the same distribution, and the \textit{alternative hypothesis} is that the reported data contains bogus data and the distributions of $g_{det}$ and $g_{ben}$ are differentiable.
We adopt the two-sample KS test \cite{massey1951kolmogorov} (see Appendix~\ref{app:two_sample_KS_test} for more details) to solve this hypothesis test problem. 
We first calculate the empirical distribution of $g_{ben}$ and $g_{det}$. Denoting the empirical distribution function of $g_{ben}$ and $g_{det}$ as $F_{ben}(x)$ and $F_{det}(x)$ respectively, we have $F_{ben}(x) = \frac{1}{m} \sum_{i=1}^{m} \mathbb{I}_{[-\infty, x]}\left( W_1(\hat{X}_2^{(i)}, \hat{X}_3^{(i)}) \right)$ and $F_{det}(x) = \frac{1}{m} \sum_{i=1}^{m} \mathbb{I}_{[-\infty, x]}\left( W_1(\hat{X}, \hat{X}_2^{(i)}) \right)$,
where $\mathbb{I}_{[-\infty, x]}(d)$ is the indicator function equal to 1 if $d \leq x$ and 0 otherwise.
Then we have the KS statistic $S = \sup_{x} \left| F_{ben}(x) - F_{det}(x) \right|$
and the $p$-value is $p = 2e^{-2S^2 \frac{m^2}{2m}}$ \cite{massey1951kolmogorov}. Comparing $p$ with the specified significance level $\alpha$, we can accept or reject the null hypothesis, i.e., whether the reported data is polluted.

\vspace{-.05in}
\begin{tcolorbox}[enhanced,drop shadow southwest, left=1pt, right=1pt, top=0pt, bottom=0pt]
For RQ-3, our answer is affirmative. We propose a zero-shot detection method by leveraging the LDP characteristics and likelihood information in data distribution for the underlying LDP protocols. Our method provides a new perspective for designing universal defense in practice without the unrealistic assumption of prior knowledge about ground truth data information. 
\end{tcolorbox}
\vspace{-.05in}

\subsection{Evaluation}
\label{sec:detection_evaluation}

\noindent\textbf{Datasets.}
We use the same datasets as attack evaluation introduced in Section~\ref{sec:attack_experiments}.

\noindent\textbf{Adaptive Attack on OUE.}
In the robustness evaluation, we focus on the most empirically effective attack to investigate the robustness of OUE in the worst case. However, the fake value that only sets the right-most bin as 1 is a clear anomaly report in LDP, which can be detected trivially. To better evaluate our detection, we also adopt a more stealthy attack strategy for OUE. In addition to the right-most bin, the attacker also randomly samples $l$ bits of the perturbed vector $\hat{\bm{y}}^{(j)}$ and sets them to 1, such that the total number of 1's is equal to the expected number of 1's in the perturbed vector of a genuine user. Based on the perturbation probability of OUE, we have $l = \lfloor \frac{m_o - 1}{e^\epsilon + 1} - \frac{1}{2} \rfloor$. We call this attack OUE-Pad in the experiments. Although OUE-Pad is more stealthy, it is weaker than the adopted attack (see Table~\ref{tab:auc_table_CFO}) and thus the result would not be a satisfactory indicator of the actual resilience performance of OUE.

\noindent\textbf{Existing Detection Methods.} To the best of our knowledge, there exists no detection for SW. 
The existing detection methods \cite{cao2021data, huang2024ldpguard} are designed for the poisoning attack on CFO protocols that aim to promote the frequencies of targeted items. In  \cite{cao2021data}, two detection methods are proposed, namely, conditional probability-based attack detection (CPAD) and malicious user detection (MUD). MUD is also adopted in \cite{huang2024ldpguard} for OUE and OLH. The CPAD cannot be applied to attacks on numerical data since it depends on two assumptions that do not hold in our attack: 1) the true frequency of the target bin is close to zero, and 2) the frequency of the target bin after attack becomes the top-$N$ frequencies. Further, it is ineffective for small $\beta$ (e.g., $\beta \leq 20\%$) as it depends on the true data distribution, which is unknown to the server. 

When the perturbed frequency closely resembles that derived from a legitimate distribution, the server cannot discern whether the noisy report is polluted. Consequently, MUD considers the perturbed result under the assumption that all users are genuine and their values fall within the targeted bin (i.e., the right-most bin). In cases where the perturbed results exhibit statistical divergence in this extreme scenario, the server can detect the attack by inferring that the perturbed result cannot be derived from a legitimate distribution. In other words, if the perturbed reports are generated through LDP protocols, it is unlikely to observe a large number of reports supporting the same targeted bin. MUD sets a threshold $\tau$ and triggers an alarm if the number of the perturbed reports supporting the right-most bin exceeds $\tau$. In \cite{cao2021data}, MUD does not work for GRR since the fake value is only a single index no matter whether or not the attacker follows the protocol. An effective threshold for OUE is determined as $\sqrt{\frac{n/4}{0.01}} + \frac{n}{2}$, while for OLH and HST, the thresholds are the smallest solutions for $\tau$ in $I(\frac{1}{2}; \tau, n - \tau + 1) \leq 0.01$ and $I(\frac{e^\epsilon}{e^\epsilon + 1}; \tau, n - \tau + 1) \leq 0.01$ respectively. Here, $I$ is the regularized incomplete beta function \cite{abramowitz1988handbook}.

\noindent\textbf{Metrics.}
The proposed detection classifies the reported result as ``polluted" or ``unpolluted".  Therefore, we can use metrics for binary classifiers to measure the performance of our mitigation. In the evaluation, we adopt the Receiver Operating Characteristic curve \cite{fawcett2006introduction} to capture the relationship between the true positive rate and false positive rate of the classifier. The Area Under the Curve is used to measure the performance of the classifier.  In general, the AUC ranges from 0 to 1. The larger the AUC is, the better the classifier performs. When AUC equals 0.5, the classifier predicts a random class for all the data points. A classifier with 100\% false prediction has an AUC of 0, while a perfect classifier that is always correct will generate an AUC of 1.0.
For the total of 100 trials of experiments for each protocol, half of the times are under attack while the other half is without attack. We also mark the unpolluted and polluted results as the positive class and negative class respectively. 

\noindent\textbf{Parameter Settings.}
We set the default $m = 10$. According to our analysis, the attack efficacy is low with large $\epsilon$. Therefore, we only evaluate the detection with small $\epsilon$, i.e., 0.2, 0.6 and 1. In order to comprehensively study the detection effectiveness, we further test $\beta$ at 10\% for additional insights and thus set the relevant parameters $\beta = 1\%$,  $2.5\%$, $5\%$ and $10\%$, $m_o = 32$ and $m_s = 512$.

\begin{table*}[!htbp]
\centering
\scriptsize
\caption{AUC values of attack detection for SW.}
\label{tab:auc_table_SW}
\begin{tabular}{lcccccccccc}
\hline
\multirow{2}{*}{Dataset} & \multirow{2}{*}{$\beta$} & \multirow{2}{*}{$\epsilon$} & \multicolumn{2}{c}{SW right-most bin} & \multicolumn{2}{c}{SW [1+2/3b, 1+b]} & \multicolumn{2}{c}{SW [1, 1+b]} & \multicolumn{2}{c}{SW [1-b, 1+b]} \\
 &  &  & AUC & \ASG & AUC & \ASG & AUC & \ASG & AUC & \ASG \\ \hline
\multicolumn{1}{l|}{} & \multicolumn{1}{l}{} & $0.2$ & 1.00 & 0.00289 & 1.00 & 0.0404 & 1.00 & 0.04536 & 0.8592 & 0.04571 \\
\multicolumn{1}{l|}{} & $1\%$ & $0.6$ & 1.00 & 0.00199 & 1.00 & 0.0224 & 1.00 & 0.0239 & 0.685 & 0.0179 \\
\multicolumn{1}{l|}{} & \multicolumn{1}{l}{} & $1$ & 1.00 & 0.00084 & 1.00 & 0.0127 & 1.00 & 0.0131 & 0.56 & 0.0122 \\ \cline{2-11} 
\multicolumn{1}{l|}{} &  & $0.2$ & 1.00 & 0.00364 & 1.00 & 0.1042 & 1.00 & 0.1132 & 0.9996 & 0.1052 \\
\multicolumn{1}{l|}{} & $2.5\%$ & $0.6$ & 1.00 & 0.00362 & 1.00 & 0.0569 & 1.00 & 0.0571 & 0.9872 & 0.0447 \\
\multicolumn{1}{l|}{} &  & $1$ & 1.00 & 0.00215 & 1.00 & 0.03129 & 1.00 & 0.0329 & 0.9840 & 0.03037 \\ \cline{2-11} 
\multicolumn{1}{l|}{} &  & $0.2$ & 1.00 & 0.01114 & 1.00 & 0.2056 & 1.00 & 0.2284 & 1.00 & 0.1915 \\
\multicolumn{1}{l|}{Taxi} & $5\%$ & $0.6$ & 1.00 & 0.00962 & 1.00 & 0.0981 & 1.00 & 0.0991 & 1.00 & 0.0891 \\
\multicolumn{1}{l|}{} &  & $1$ & 1.00 & 0.00807 & 1.00 & 0.0614 & 1.00 & 0.0644 & 1.00 & 0.0601 \\ \cline{2-11} 
\multicolumn{1}{l|}{} &  & $0.2$ & 1.00 & 0.0171 & 1.00 & 0.3155 & 1.00 & 0.3416 & 1.00 & 0.2615 \\
\multicolumn{1}{l|}{} & $7.5\%$ & $0.6$ & 1.00 & 0.0213 & 1.00 & 0.1324 & 1.00 & 0.1352 & 1.00 & 0.1286 \\
\multicolumn{1}{l|}{} &  & $1$ & 1.00 & 0.0176 & 1.00 & 0.091 & 1.00 & 0.0959 & 1.00 & 0.0896 \\ \cline{2-11} 
\multicolumn{1}{l|}{} &  & $0.2$ & 1.00 & 0.02981 & 1.00 & 0.3859 & 1.00 & 0.3707 & 1.00 & 0.3145 \\
\multicolumn{1}{l|}{} & $10\%$ & $0.6$ & 1.00 & 0.0363 & 1.00 & 0.1642 & 1.00 & 0.1696 & 1.00 & 0.1652 \\
\multicolumn{1}{l|}{} & \multicolumn{1}{l}{} & $1$ & 1.00 & 0.02907 & 1.00 & 0.1218 & 1.00 & 0.1269 & 1.00 & 0.1186 \\ \hline
\end{tabular}
\end{table*}

\begin{table*}[!htbp]
\setlength\tabcolsep{1pt}
\centering
\scriptsize
\caption{AUC values of attack detection for CFO-based mechanisms. ``--'' indicates the AUC value is zero.}
\label{tab:auc_table_CFO}
\begin{tabular}{lcccccccccccccccc}
\hline
\multirow{2}{*}{Dataset} & \multirow{2}{*}{$\beta$} & \multirow{2}{*}{$\epsilon$} & \multicolumn{2}{c}{HST-Server} & \multicolumn{2}{c}{HST-User} & \multicolumn{2}{c}{OLH-Server} & \multicolumn{2}{c}{OLH-User} & \multicolumn{2}{c}{OUE} & \multicolumn{2}{c}{GRR} & \multicolumn{2}{c}{OUE-Pad} \\
 &  &  & Ours~/~MUD & \ASG & Ours~/~MUD & \ASG & Ours~/~MUD & \ASG & Ours~/~MUD & \ASG & Ours~/~MUD & \ASG & Ours~/~MUD & \ASG & Ours~/~MUD & \ASG \\ \hline
\multicolumn{1}{l|}{} & \multicolumn{1}{l}{} & $0.2$ & 0.4416~/~-- & 0.0484 & 1.00~/~-- & 0.391 & 0.476~/~-- & 0.047 & 0.9952~/~-- & 0.207 & 1.00~/~-- & 0.391 & 1.00~/~-- & 0.391 & 1.00~/~-- & 0.094 \\
\multicolumn{1}{l|}{} & $1\%$ & $0.6$ & 0.4384~/~-- & 0.0054 & 1.00~/~-- & 0.101 & 0.4232~/~-- & 0.0189 & 0.6224~/~-- & 0.105 & 1.00~/~-- & 0.071 & 1.00~/~-- & 0.184 & 1.00~/~-- & 0.021 \\
\multicolumn{1}{l|}{} & \multicolumn{1}{l}{} & $1$ & 0.3306~/~-- & 0.008 & 0.9744~/~-- & 0.058 & 0.392~/~-- & 0.015 & 0.5784~/~-- & 0.063 & 0.6808~/~-- & 0.026 & 0.8992~/~-- & 0.096 & 0.6407~/~-- & 0.016 \\ \cline{2-17} 
\multicolumn{1}{l|}{} &  & $0.2$ & 0.4972~/~-- & 0.0992 & 1.00~/~-- & 0.391 & 0.6167~/~-- & 0.117 & 1.00~/~-- & 0.285 & 1.00~/~-- & 0.39 & 1.00~/~-- & 0.391 & 1.00~/~-- & 0.195 \\
\multicolumn{1}{l|}{} & $2.5\%$ & $0.6$ & 0.4872~/~-- & 0.03432 & 1.00~/~-- & 0.392 & 0.4933~/~-- & 0.0407 & 1.00~/~-- & 0.2012 & 1.00~/~-- & 0.382 & 1.00~/~-- & 0.387 & 1.00~/~-- & 0.063 \\
\multicolumn{1}{l|}{} &  & $1$ & 0.3696~/~-- & 0.027 & 1.00~/~-- & 0.344 & 0.4~/~-- & 0.035 & 1.00~/~-- & 0.146 & 0.9696~/~-- & 0.11 & 1.00~/~-- & 0.22 & 0.913~/~-- & 0.044 \\ \cline{2-17} 
\multicolumn{1}{l|}{} &  & $0.2$ & 0.555~/~-- & 0.18 & 1.00~/~-- & 0.3912 & 0.6504~/~-- & 0.204 & 1.00~/~-- & 0.3353 & 1.00~/~-- & 0.391 & 1.00~/~-- & 0.391 & 1.00~/~-- & 0.391 \\
\multicolumn{1}{l|}{Taxi} & $5\%$ & $0.6$ & 0.5352~/~-- & 0.072 & 1.00~/~-- & 0.391 & 0.5392~/~-- & 0.106 & 1.00~/~-- & 0.277 & 1.00~/~-- & 0.391 & 1.00~/~-- & 0.391 & 1.00~/~-- & 0.147 \\
\multicolumn{1}{l|}{} &  & $1$ & 0.4976~/~-- & 0.048 & 1.00~/~-- & 0.389 & 0.4844~/~-- & 0.0771 & 1.00~/~-- & 0.23 & 1.00~/~-- & 0.389 & 1.00~/~-- & 0.385 & 1.00~/~-- & 0.102 \\ \cline{2-17}
\multicolumn{1}{l|}{} &  & $0.2$ & 0.6211~/~-- & 0.29 & 1.00~/~-- & 0.391 & 0.7411~/~-- & 0.31 & 1.00~/~-- & 0.357 & 1.00~/~-- & 0.392 & 1.00~/~-- & 0.391 & 1.00~/~-- & 0.391 \\
\multicolumn{1}{l|}{} & $7.5\%$ & $0.6$ & 0.6089~/~-- & 0.108 & 1.00~/~-- & 0.39 & 0.6494~/~-- & 0.15 & 1.00~/~-- & 0.31 & 1.00~/~-- & 0.391 & 1.00~/~-- & 0.391 & 1.00~/~-- & 0.255 \\
\multicolumn{1}{l|}{} &  & $1$ & 0.5667~/~-- & 0.073 & 1.00~/~-- & 0.391 & 0.5583~/~-- & 0.11 & 1.00~/~-- & 0.27 & 1.00~/~-- & 0.391 & 1.00~/~-- & 0.391 & 1.00~/~-- & 0.154 \\ \cline{2-17}
\multicolumn{1}{l|}{} &  & $0.2$ & 0.8432~/~0.575 & 0.382 & 1.00~/~0.575 & 0.392 & 0.7912~/~0.55 & 0.377 & 1.00~/~0.55 & 0.37 & 1.00~/~0.55 & 0.392 & 1.00~/~-- & 0.391 & 1.00~/~0.55 & 0.391 \\
\multicolumn{1}{l|}{} & $10\%$ & $0.6$ & 0.62~/~-- & 0.15 & 1.00~/~-- & 0.39 & 0.63~/~-- & 0.207 & 1.00~/~-- & 0.335 & 1.00~/~-- & 0.39 & 1.00~/~-- & 0.391 & 1.00~/~-- & 0.391 \\
\multicolumn{1}{l|}{} & \multicolumn{1}{l}{} & $1$ & 0.5072~/~-- & 0.097 & 1.00~/~-- & 0.39 & 0.5338~/~-- & 0.15 & 1.00~/~-- & 0.303 & 1.00~/~-- & 0.391 & 1.00~/~-- & 0.391 & 1.00~/~-- & 0.248 \\ \hline
\end{tabular}
\vspace{-10pt}
\end{table*}

\subsubsection{Detection Results}
We have similar observations across all three datasets and only show AUC values of detection for SW and CFOs with the dataset \textit{Taxi} in Table~\ref{tab:auc_table_SW} and Table~\ref{tab:auc_table_CFO} respectively due to space limit. The rest are shown in Table~\ref{tab:full_auc_table_SW} and \ref{tab:full_auc_table_CFO} in Appendix~\ref{app:detection_supplementary_result}.
In addition, the defender often cares about the overall impact of the attack in practice. Thus, we also compute the average \ASG of
the attack to show the absolute shift of the distribution. We have the following key observations.

\begin{itemize}[leftmargin=*]
    \item Overall, the results demonstrate our detection is effective.
    \begin{itemize}[leftmargin=*]
        \item For SW, all AUC values are larger than 0.95 for $\beta \geq 2.5\%$, which demonstrates an excellent detection rate by generating the benchmark and capturing the deviation of polluted results from the benchmark. The AUC drops when $\beta$ is small (e.g., $\beta = 1\%$). However, the detection still works well for attacks on the right-most bin, on the range $[1+\frac{2}{3}b, 1+b]$ and range $[1, 1+b]$. This is because the fake values in these attacks are concentrated on a small region, leading to an obvious statistical abnormality.

        \item For CFO protocols, our detection significantly outperforms MUD. MUD fails to detect the attack except for only large $\beta$ and small $\epsilon$ but with a very constrained detection performance (AUC $\leq 0.575$). On the contrary, our method performs consistently across various settings of $\beta$ and $\epsilon$. In particular, for GRR, OUE and CFOs in the \textit{User} setting, the AUC is over 0.96 in most cases. 
    \end{itemize}

    \item The AUC of GRR, OUE, OUE-Pad and other CFO protocols under \textit{User} setting is significantly larger than that of CFO protocols under \textit{Server} setting. This is because the attacker can only forge the value but cannot control the hash function in \textit{Server} setting, thus limited information available for detection. Our detection can still identify the attack in the \textit{Server} setting with large $\beta$ (e.g., $\beta = 10\%$).

    \item By the extra padded bits, OUE-Pad achieves lower \ASG than OUE, but our detection still has high AUC on OUE-Pad. This is because OUE-Pad does not concentrate on the right-most bin and the padding bits is uniformly sampled instead of following OUE perturbation probability and thus the recovered data cannot reproduce the polluted report, leading to a distinction between $g_{det}$ and $g_{ben}$ in detection.
    
    \item MUD performs poorly and can only detect the attack with small $\epsilon = 0.2$ and large $\beta = 10\%$. This is because it compares the noisy report with the extreme case and thus the threshold $\tau$ is too large to detect the attack for small $\beta$ and large $\epsilon$.
    
    \item The AUC of the detection against the attack on $[1-b, 1+b]$ in SW is lower than those adopting other attack strategies, especially for small $\beta$. The reason is that the distribution of fake values in this range is close to the output probability distribution of SW. Therefore, this type of attack is closer to the baseline and its AUC is the smallest.

    \item The AUC for OLH-User is lower than that for HST-User, GRR, OUE and OUE-Pad. This is because fake values in HST-User, GRR and OUE only support the right-most bin and OUE-Pad contains dummy bits following a different distribution from OUE perturbation probability, leading to higher statistical anomaly than in OLH-User where the crafted values support a set of bins at the right end.
    
    \item As $\epsilon$ grows, \ASG reduces and the detection performs worse. This is because larger $\epsilon$ results in less noise, thus making the attack closer to the baseline.
    
    \item The \ASG and AUC are getting greater with increased $\beta$ because more compromised users facilitate the distribution manipulation, thus increasing the distance between the reported bogus data and the reconstructed result, which deviates farther from the benchmark.
    
    \item In general, higher \ASG leads to higher AUC but AUC is not solely determined by the attack gain. It is also correlated with $\beta$ and $\epsilon$ at the same time. Given a fixed \ASG, we observe that the detection rate can be different. For example, the \ASG of attack range $[1-b, 1+b]$ is about 0.045 under $\epsilon = 0.2, \beta = 1\%$ and $\epsilon = 0.6, \beta = 2.5\%$ on dataset \textit{Taxi}. However, the AUC is 0.9872 for $\beta = 2.5\%$ but is only 0.8592 for $\beta = 1\%$.
    
    \item Overall, the AUC with dataset \textit{Taxi} is higher than the other two datasets. This is because, on \textit{Taxi}, the probability density of the region near the right end of the domain is high. Thus, the frequencies of polluted bins are still high even after averaging in EMS, leading to higher anomaly.
\end{itemize}


\begin{figure*}[!htbp]
    \centering
    \includegraphics[scale=0.35]{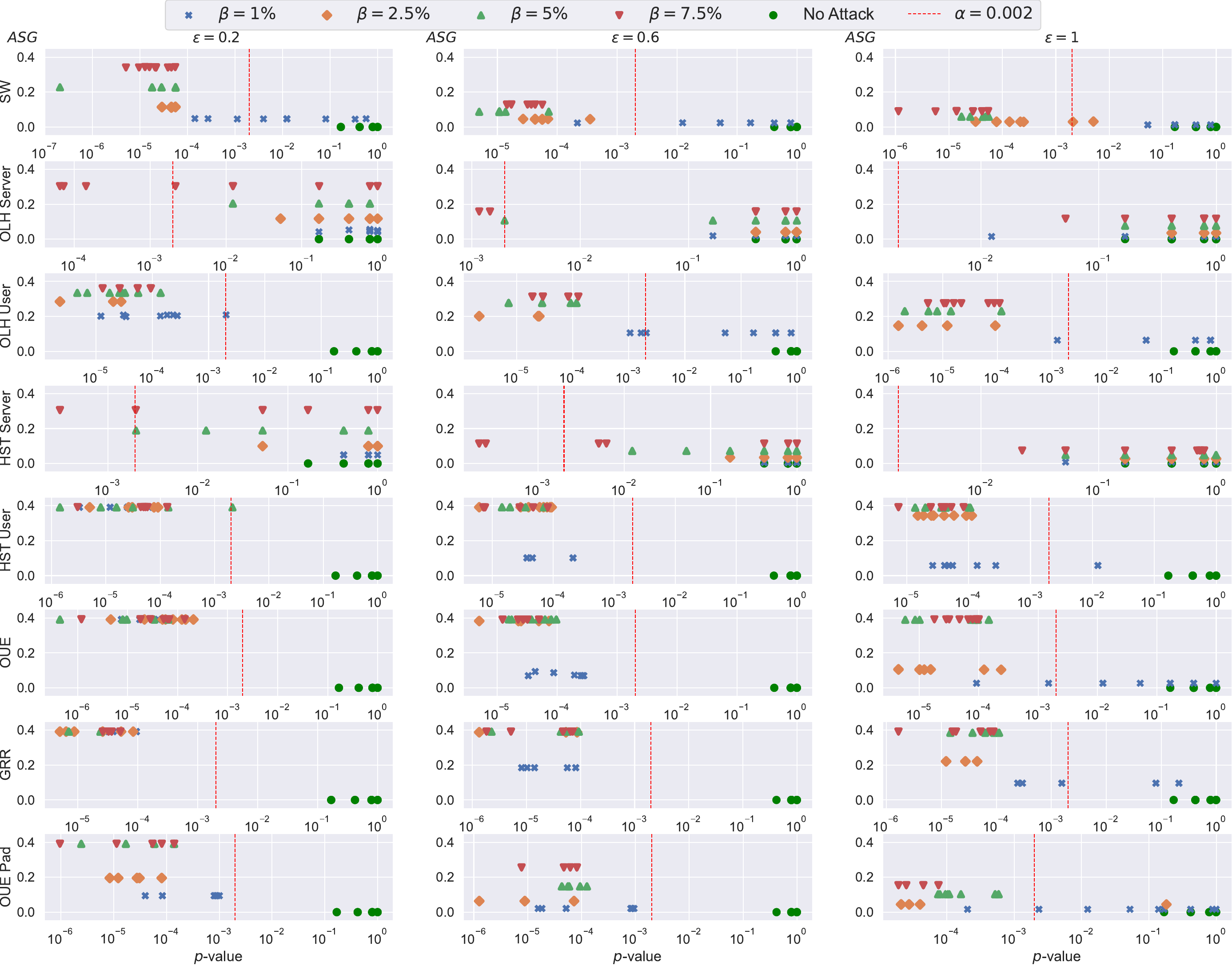}
    \caption{Relationship between shift gain and $p$-value with \textit{Taxi}. Each row represents a protocol with varying $\epsilon$.} 
    \label{fig:pvalue_mean_gain}
    \vspace{-10pt}
\end{figure*}

\subsubsection{Relationship between \ASG and $p$-value}

Despite the superior overall performance of our detection compared to the existing method, the results show that the proposed method cannot effectively detect the attack for small $\beta$ on SW and does not achieve high AUC for CFO in the \textit{Server} setting. In this subsection, we study the attack efficacy versus detection performance. To this end, we analyze the relationship between the \ASG and the $p$-value in the detection. To achieve a low false positive rate, we set the $\alpha = 0.002$. We also select the best attack strategy with the highest \ASG/\SGR for SW under each parameter setting. By repeating the attack 10 times, we use the scatterplot to show the \ASG distribution as the function of $p$-value in detection. We have similar observations across all three datasets and thus only show the results with \textit{Taxi} in Figure~\ref{fig:pvalue_mean_gain}.  
\begin{itemize}[leftmargin=*]
    \item For SW, there are two key observations:
    \begin{itemize}[leftmargin=*]
        \item When the number of malicious users is sufficient (e.g., $\beta \geq 5\%$) to contribute to a high \ASG, the $p$-value in detection is significantly smaller than that of \textit{No Attack}. This is because a large number of fake values give rise to significant anomalies.
    
        \item When the attack bypasses the detection (i.e., the $p$-value is greater than $\alpha$) with a small $\beta$, the \ASG is close to \textit{No Attack} because of the diminished attack effectiveness.
    \end{itemize}

    \item For CFO protocols:
    \begin{itemize}[leftmargin=*]
        \item With HST and OLH in the \textit{User} setting, GRR, OUE and OUE-Pad, the attack makes a limited impact on the final estimate (e.g., $\ASG \leq 0.1$) when it evades the detection.

        \item The attack in the \textit{Server} setting tends to evade the detection with non-negligible \ASG. This is because the attacker can only manipulate limited information. On the flip side, the lack of sufficient information prevents identifying the attack effectively. As a result, it is challenging to detect the data manipulation in the \textit{Server} setting.
    \end{itemize}
\end{itemize}

\noindent\textbf{Why our detection is better than MUD?}
Our detection method utilizes the distribution information to construct an effective benchmark, allowing precise identification of skewness in noisy reports. In contrast, MUD relies solely on checking the frequencies of specific bins. It is effective when the attacker targets a set of bins since it is unlikely that these bins are all supported by a group of users. However, our attack on CFO protocols targets only the right-most bin, leading to statistically insignificant anomalies. Consequently, MUD must employ a large threshold to mitigate false positives, which requires control over an exceptionally large user base and results in poor performance for small $\beta$.


In the \textit{Server} setting, the attacker can only forge values and lacks control over the hash function. Detection can only rely on the frequencies of specific bins. However, our approach utilizes the distribution information across all bins, creating a robust detection benchmark. Unlike MUD, our method exploits correlations and inherent statistical properties beyond bin frequencies, thus achieving better detection.

\noindent\textbf{Summary.}
Our evaluation unveils a bittersweet result for CFO-based protocols. For GRR, OUE and CFOs in the \textit{User} setting, they are naturally more vulnerable to the data poisoning attack. However, the high attack anomaly with these protocols also enables our detection to effectively identify the threat. On the other hand, CFOs in the \textit{Server} setting are more robust to data manipulation but suffer from a low detection rate. In contrast, SW shows good resilience against the attack and offers rich distribution information to facilitate the detection. Meanwhile, prior research shows that better utility for numerical data is obtained with SW compared to CFO-based protocols \cite{li2020estimating}. Therefore, SW would be a preferable protocol with balanced security, privacy, and utility expectations. We expect our findings will still hold in the scenarios where the attack is launched towards other regions of the distribution since the attack intuition remains to increase the distribution density within the target region.




\section{Discussion and Future Direction}

\noindent\textbf{ASG/SGR for LDP Robustness Improvement.}
Security has become a critical factor when we evaluate LDP protocols in addition to utility assessment. ASG and SGR as stronger indicators of attack resilience can be adopted in concert with utility metrics, such as mean square error (MSE) to offer a more comprehensive evaluation of LDP in a hostile environment. Further, ASG/SGR can provide insights into robust LDP design. For example, the optimal number $g$ in OLH \cite{wang2017locally} is determined by minimizing the MSE of LDP noise while our analysis shows that $g$ also has an impact on the protocol robustness. As a result, we may further  design a more robust OLH by finding a $g$ that minimizes attack efficacy informed by ASG/SGR along with the utility optimization.




\noindent\textbf{Limitation of Our Detection and Alternative Mitigation.}
The detection needs $m$ noisy results to produce $g_{det}$ and $g_{ben}$, which may incur a time cost that does not favor real-time applications, especially when a large number of users are involved. The proposed detection method may not be directly applicable to categorical data since it is built on Wasserstein distance for numerical settings.



In addition to detection, attack recovery is another important defense strategy, especially when the threat is persistent or recollecting data is impossible \cite{li2022fine, du2023differential, sun2024ldprecover, huang2024ldpguard}. In this case, accurate detection can provide rich information to help restore damaged utility by suppressing the attack impact on the result. We leave this to our future work.


\noindent\textbf{Robustness of Shuffler-based LDP.}
Shuffler may further improve robustness.
Shuffler-based LDP protocols \cite{balle2019privacy, wang13improving}  deploy a shuffler to break the link between users and their reports. Intuitively, this anonymity can provide a better privacy-utility trade-off because users can add less noise while achieving the same level of privacy. Therefore, the results recovered by aggregation are close to the input values, and the attack efficacy could be close to the baseline attack. Our detection framework could apply to shuffler-based protocols for numerical data \cite{cheu2021differentially}. Although the shuffler conceals the user's identity, it does not lose the distribution information. Therefore, the server can still recover the data distribution from the shuffled noisy reports and reconstruct the noisy results to capture the attack.





\noindent\textbf{Attack Detection for Graph Statistics Estimation.}
In graph data mining, LDP \cite{ye2020lf} focuses on estimating the modularity, structural similarity, etc. Node degree and adjacency bit vector are two atomic graph metrics. To manipulate the statistics of a graph, the attacker may inject a false adjacency bit vector to pollute the node degree and adjacent matrix, and further manipulate the statistics. The server can derive the degree distribution from the adjacent matrix and our detection could synthesize a graph by graph generation model \cite{qin2017generating}. By performing a reconstruction and a hypothesis test, the attack is likely to be detected.

\section{Conclusion}


Studying the robustness of LDP protocols under data poisoning attacks is a critical first step toward restoring the security and reliability of LDP protocols in practice. In this paper, we investigate the impact of malicious data manipulation on the state-of-the-art LDP protocols for numerical data estimation and find that not all of these protocols are equally robust against the concerned threats. The \textit{Server} setting for CFO with binning and consistency provides better robustness than the \textit{User} setting. CFO-\textit{Server} along with SW is the most robust under our data poisoning attack while SW is preferable with better detection sensitivity and utility. Our research further advances the prior knowledge and reveals new relationships between LDP design choices and robustness. 

\section{Acknowledgment}
We would like to thank the Shepherd and anonymous reviewers for their insightful comments and guidance. Wenhai Sun was supported in part by NSF grant CNS-2238680.

\bibliography{ref}
\bibliographystyle{IEEEtran}
\begin{appendices}
\begin{figure*}[!htb]
    \centering
    \includegraphics[scale = 0.35]{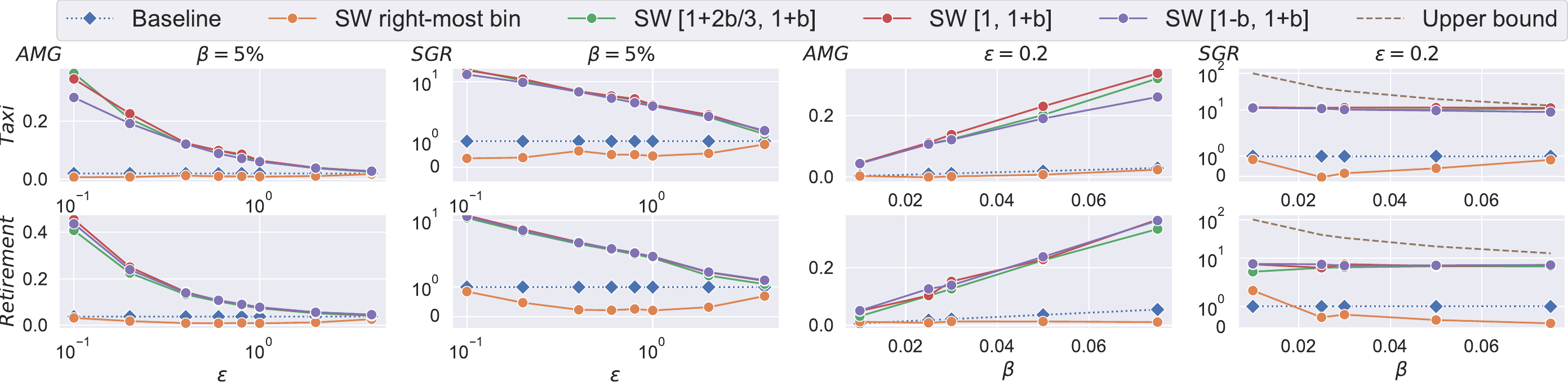}
    \caption{Attack results for SW mechanism with \textit{Taxi} and \textit{Retirement} datasets, varying $\epsilon$ from $0.1$ to $4$. The fake values are injected into 1) the right-most bin, 2) range $[1+\frac{2b}{3}, 1+b]$, 3) range $[1, 1+b]$ and 4) range $[1-b, 1+b]$.}
    \label{fig:SW_bin_eps_beta_full_results}
\end{figure*}

\section{Proof of Theorem \ref{the:security_g_relationship}} \label{app:proof_security_g_relationship}
\begin{proof}
    According to \cite{wang2017locally}, HST and OLH are equivalent for $g=2$. Therefore we only study the OLH and the results can be applied to HST. \SGR is proportional to \ASG and thus we only study \ASG and the conclusion is also applicable to \SGR. Since local hashing-based mechanisms can be implemented by two settings, we first study the \textit{User} case. Ideally, the fake report $\langle \hat{H}, \hat{y} \rangle$ contributes only to the estimated frequency of the largest value in the $m_o$-th bin. To study the relationship between $g$ and \ASG, we need to analyze the expected shift gain after the attack. The idea is to study frequency change on each bin after the attack and aggregate them to obtain the shift gain. The term $C(B_i)$ counts the number of noisy data falling in the $i$-th bin $B_i$, and it can also be presented by indicating function as $\sum_{j=1}^{n} I_{y^{(j)}}(i)$, where $I_{y^{(j)}}(i)$ equals 1 if $H^{(j)}(i) = \hat{x}_b^{(j)}$. Denote the true frequency of the $i$-th bin as $f_i$, the excepted frequency change $\Delta f_i$ for the $i$-th ($i \neq m_o$) bin is
    \begin{align*}
        &\mathbb{E}(\Delta f_i) = \mathbb{E} \Bigg[ \frac{\sum_{j=1}^{(1-\beta)n} I_{y^{(j)}}(i) + \sum_{j=1}^{\beta n} I_{\hat{y}^{(j)}}(i) - nq}{n(p-q)} \\
        &- \frac{\sum_{j=1}^{n} I_{y^{(j)}}(i) - nq}{n(p-q)} \Bigg] \\
        &= \frac{1}{n(p-q)} \mathbb{E} \left[ n(q-p)f_i \beta - q \beta \right] = -f_i \beta - \frac{\beta}{g-2}
    \end{align*}
    The second equality is because for $i \neq m_o$, the $I_{y^{(j)}}(i)$ is 1 with probability $p$ and 0 with probability $q$, and the $I_{\hat{y}^{(j)}}(i)$ for compromised user is 0 ideally.
    For the $m_o$-th bin, the expected frequency change is
    \begin{align*}
        &\mathbb{E}(\Delta f_{m_o}) \\
        &= \mathbb{E} \left[ \frac{\sum_{j=1}^{(1-\beta)n} I_{y^{(j)}}(i) + \sum_{j=1}^{\beta n} 1 - nq}{n(p-q)} - \frac{\sum_{j=1}^{n} I_{y^{(j)}}(i) - nq}{n(p-q)} \right] \\
        &= \frac{\beta}{p-q} - \beta f_{m_o} - \frac{\beta}{(p-q) g}
    \end{align*}
    Thus the expected \ASG is $\sum_{v=1}^{m_o} \sum_{i=1}^{v} \mathbb{E}(\Delta f_i)$, which is the linear combination of $\mathbb{E}(\Delta f_i) (\forall i \in [m_o])$ and is proportional to $-\frac{1}{g}$. Because the derivative of every $\mathbb{E}(\Delta f_i)$ with respect to $g$ is positive, $\sum_{v=1}^{m_o} \sum_{i=1}^{v} \mathbb{E}(\Delta f_i)$ increases as $g$ grows.
    
    The proof is the same for \textit{Server} setting except that the term $\sum_{j=1}^{\beta n} I_{\hat{y}^{(j)}}(i)$ is 0 with probability $\frac{g-1}{g}$ while 1 with probability $\frac{1}{g}$.
    This is because the hash is picked by the server and each hash maps input to domain $[g]$ uniformly at random. Then we have the expected shift gain as
    \begin{align*}
        &\mathbb{E}(\Delta f_i) = -\beta f_i \quad (i \neq m_o) \\
        &\mathbb{E}(\Delta f_{m_o}) = \frac{\beta}{p-q} - \beta f_{m_o} - \frac{\beta}{(p-q) g}
    \end{align*}
    Since the derivative of every $\mathbb{E}(\Delta f_i)$ with respect to $g$ is still positive, $\sum_{v=1}^{m_o} \sum_{i=1}^{v} \mathbb{E}(\Delta f_i)$ also increases as $g$ grows under \textit{Server} setting.
\end{proof}

\section{Data Distribution and Supplemental Attack Results}
\label{app:supplemental_experiment_results}



The data distribution is shown in Figure~\ref{fig:data_dist}. As the supplemental experimental results of Figure~\ref{fig:SW_bin_eps_beta} in Section~\ref{sec:attack_experiments}, Figure~\ref{fig:SW_bin_eps_beta_full_results} shows the attack results on SW mechanism with the rest two datasets \textit{Taxi} and \textit{Retirement}. We observe a similar result to that with dataset $\mathcal{N}(0, 10)$. Except for injecting into the right-most bin, injecting fake values into the range $[1+\frac{2b}{3}, 1+b]$, range $[1, 1+b]$ and range $[1-b, 1+b]$ perform similarly since they all increase the probability density region of the right end of the domain.
\begin{figure}
    \centering
    \includegraphics[scale = 0.33]{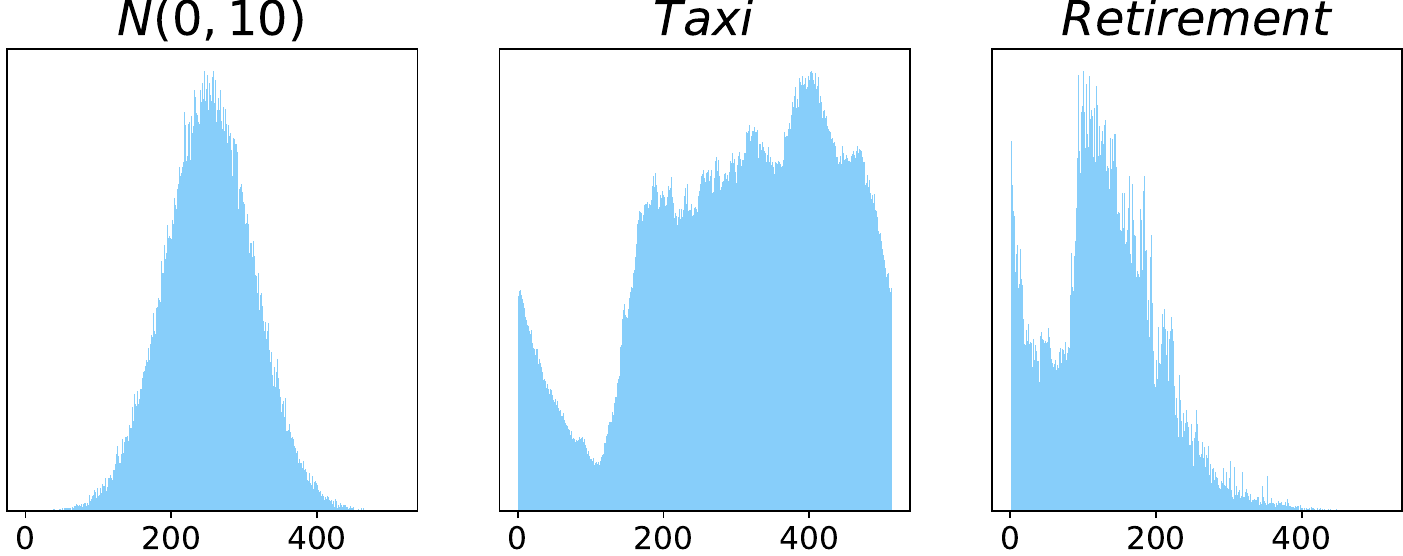}
    \caption{Datasets Distribution}
    \label{fig:data_dist}
\end{figure}







\section{Two-sample Kolmogorov–Smirnov Test} \label{app:two_sample_KS_test}
Two-sample Kolmogorov–Smirnov (KS) test \cite{massey1951kolmogorov} is used to test if two groups of samples follow the same distribution. It claims two statistical hypotheses. The \textit{null hypothesis} $\mathcal{H}_0$ states that the two groups of samples follow the same distribution, while the \textit{alternative hypothesis} $\mathcal{H}_1$ claims that the distributions of two groups are different.

Two-sample KS test has Type \uppercase\expandafter{\romannumeral1} error and Type \uppercase\expandafter{\romannumeral2} error. Type \uppercase\expandafter{\romannumeral1} error occurs when the $\mathcal{H}_0$ is true but the test mistakenly rejects it. Type \uppercase\expandafter{\romannumeral2} occurs if the $\mathcal{H}_0$ is falsely accepted when it is not true. In most problems, Type \uppercase\expandafter{\romannumeral1} error is important \cite{casella2021statistical}, and analyzers usually require the probability of Type \uppercase\expandafter{\romannumeral1} error be at most the specified small \textit{significance level} $\alpha~(0 \leq \alpha \leq 1)$. 
In order to allow users to control the error probability, the KS test returns a probabilistic estimate \textit{p-value} $p$ to measure how likely the null hypothesis is true.
The null hypothesis will be rejected if $p$ is less than the specified $\alpha$.

Two-sample KS test has three primary steps. Let $X_1 = [x_1^{(i)}]_{i=1}^{n}$ and $X_2 = [x_2^{(i)}]_{i=1}^{m}$ be two groups of samples to be tested. KS test first calculates the empirical cumulative distribution function $F_{X_1}(x)$ and $F_{X_2}(x)$ of $X_1$ and $X_2$. Then KS test computes the test statistic $S = \sup_{x} \left| F_{X_1}(x) - F_{X_2}(x) \right|$ to describe how much the test samples differ from the null hypothesis. 
Given the statistic $S$, KS test can finally derive $p$-value as $p = 2e^{-2S^2 \frac{mn}{m+n}}$~\cite{massey1951kolmogorov} and reject $\mathcal{H}_0$ if $p \leq \alpha$.

\section{Detection Details and Supplemental Results} \label{app:detection_supplementary_result}

The full detection results of SW and CFOs are shown in Table~\ref{tab:full_auc_table_SW} and \ref{tab:full_auc_table_CFO} below.

\begin{table*}[!htbp]
\centering
\scriptsize
\caption{AUC Values of Attack Detection on SW}
\label{tab:full_auc_table_SW}
\begin{tabular}{lcccccccccc}
\hline
\multirow{2}{*}{Dataset} & \multirow{2}{*}{$\beta$} & \multirow{2}{*}{$\epsilon$} & \multicolumn{2}{c}{SW right-most bin} & \multicolumn{2}{c}{SW [1+2/3b, 1+b]} & \multicolumn{2}{c}{SW [1, 1+b]} & \multicolumn{2}{c}{SW [1-b, 1+b]} \\
 &  &  & AUC & \ASG & AUC & \ASG & AUC & \ASG & AUC & \ASG \\ \hline
\multicolumn{1}{l|}{} & \multicolumn{1}{l}{} & $0.2$ & 1.00 & 0.00502 & 0.9968 & 0.0177 & 0.9048 & 0.0414 & 0.6168 & 0.0361 \\
\multicolumn{1}{l|}{} & $1\%$ & $0.6$ & 1.00 & 0.00075 & 0.9968 & 0.0141 & 0.8984 & 0.0146 & 0.5952 & 0.0134 \\
\multicolumn{1}{l|}{} & \multicolumn{1}{l}{} & $1$ & 1.00 & 0.00039 & 0.9888 & 0.0107 & 0.8368 & 0.0126 & 0.5928 & 0.0116 \\ \cline{2-11} 
\multicolumn{1}{l|}{} &  & $0.2$ & 1.00 & 0.00329 & 1.00 & 0.0996 & 1.00 & 0.1049 & 0.9940 & 0.077 \\
\multicolumn{1}{l|}{} & $2.5\%$ & $0.6$ & 1.00 & 0.000103 & 1.00 & 0.0415 & 1.00 & 0.0418 & 0.9718 & 0.0376 \\
\multicolumn{1}{l|}{} &  & $1$ & 1.00 & -0.00036 & 1.00 & 0.0294 & 0.9996 & 0.0301 & 0.9574 & 0.0269 \\ \cline{2-11} 
\multicolumn{1}{l|}{} &  & $0.2$ & 1.00 & 0.00255 & 1.00 & 0.2117 & 1.00 & 0.2272 & 0.9972 & 0.1685 \\
\multicolumn{1}{l|}{$N(0, 10)$} & $5\%$ & $0.6$ & 1.00 & 0.00177 & 1.00 & 0.0751 & 1.00 & 0.0813 & 0.9970 & 0.0764 \\
\multicolumn{1}{l|}{} &  & $1$ & 1.00 & 0.00595 & 1.00 & 0.058 & 1.00 & 0.0608 & 0.9900 & 0.05468 \\ \cline{2-11}
\multicolumn{1}{l|}{} &  & $0.2$ & 1.00 & 0.00951 & 1.00 & 0.3126 & 1.00 & 0.3297 & 0.9994 & 0.2396 \\
\multicolumn{1}{l|}{} & $7.5\%$ & $0.6$ & 1.00 & 0.00923 & 1.00 & 0.1141 & 1.00 & 0.1197 & 0.9983 & 0.1146 \\
\multicolumn{1}{l|}{} &  & $1$ & 1.00 & 0.0146 & 1.00 & 0.0832 & 1.00 & 0.0895 & 0.9976 & 0.08205 \\ \cline{2-11}
\multicolumn{1}{l|}{} &  & $0.2$ & 1.00 & 0.0102 & 1.00 & 0.4259 & 1.00 & 0.4517 & 1.00 & 0.3269 \\
\multicolumn{1}{l|}{} & $10\%$ & $0.6$ & 1.00 & 0.02301 & 1.00 & 0.1527 & 1.00 & 0.1559 & 1.00 & 0.1487 \\
\multicolumn{1}{l|}{} & \multicolumn{1}{l}{} & $1$ & 1.00 & 0.0237 & 1.00 & 0.1099 & 1.00 & 0.1158 & 1.00 & 0.1095 \\ \hline
\multicolumn{1}{l|}{} & \multicolumn{1}{l}{} & $0.2$ & 1.00 & 0.00289 & 1.00 & 0.0404 & 1.00 & 0.04536 & 0.8592 & 0.04571 \\
\multicolumn{1}{l|}{} & $1\%$ & $0.6$ & 1.00 & 0.00199 & 1.00 & 0.0224 & 1.00 & 0.0239 & 0.685 & 0.0179 \\
\multicolumn{1}{l|}{} & \multicolumn{1}{l}{} & $1$ & 1.00 & 0.00084 & 1.00 & 0.0127 & 1.00 & 0.0131 & 0.56 & 0.0122 \\ \cline{2-11} 
\multicolumn{1}{l|}{} &  & $0.2$ & 1.00 & 0.00364 & 1.00 & 0.1042 & 1.00 & 0.1132 & 0.9996 & 0.1052 \\
\multicolumn{1}{l|}{} & $2.5\%$ & $0.6$ & 1.00 & 0.00362 & 1.00 & 0.0569 & 1.00 & 0.0571 & 0.9872 & 0.0447 \\
\multicolumn{1}{l|}{} &  & $1$ & 1.00 & 0.00215 & 1.00 & 0.03129 & 1.00 & 0.0329 & 0.9840 & 0.03037 \\ \cline{2-11} 
\multicolumn{1}{l|}{} &  & $0.2$ & 1.00 & 0.01114 & 1.00 & 0.2056 & 1.00 & 0.2284 & 1.00 & 0.1915 \\
\multicolumn{1}{l|}{Taxi} & $5\%$ & $0.6$ & 1.00 & 0.00962 & 1.00 & 0.0981 & 1.00 & 0.0991 & 1.00 & 0.0891 \\
\multicolumn{1}{l|}{} &  & $1$ & 1.00 & 0.00807 & 1.00 & 0.0614 & 1.00 & 0.0644 & 1.00 & 0.0601 \\ \cline{2-11} 
\multicolumn{1}{l|}{} &  & $0.2$ & 1.00 & 0.0171 & 1.00 & 0.3155 & 1.00 & 0.3416 & 1.00 & 0.2615 \\
\multicolumn{1}{l|}{} & $7.5\%$ & $0.6$ & 1.00 & 0.0213 & 1.00 & 0.1324 & 1.00 & 0.1352 & 1.00 & 0.1286 \\
\multicolumn{1}{l|}{} &  & $1$ & 1.00 & 0.0176 & 1.00 & 0.091 & 1.00 & 0.0959 & 1.00 & 0.0896 \\ \cline{2-11} 
\multicolumn{1}{l|}{} &  & $0.2$ & 1.00 & 0.02981 & 1.00 & 0.3859 & 1.00 & 0.3707 & 1.00 & 0.3145 \\
\multicolumn{1}{l|}{} & $10\%$ & $0.6$ & 1.00 & 0.0363 & 1.00 & 0.1642 & 1.00 & 0.1696 & 1.00 & 0.1652 \\
\multicolumn{1}{l|}{} & \multicolumn{1}{l}{} & $1$ & 1.00 & 0.02907 & 1.00 & 0.1218 & 1.00 & 0.1269 & 1.00 & 0.1186 \\ \hline
\multicolumn{1}{l|}{} & \multicolumn{1}{l}{} & $0.2$ & 1.00 & 0.01242 & 1.00 & 0.02808 & 0.9856 & 0.0441 & 0.6256 & 0.043 \\
\multicolumn{1}{l|}{} & $1\%$ & $0.6$ & 1.00 & 0.00356 & 1.00 & 0.0234 & 0.985 & 0.0209 & 0.605 & 0.0198 \\
\multicolumn{1}{l|}{} & \multicolumn{1}{l}{} & $1$ & 1.00 & 0.00174 & 1.00 & 0.0123 & 0.9816 & 0.0167 & 0.5624 & 0.0157 \\ \cline{2-11} 
\multicolumn{1}{l|}{} &  & $0.2$ & 1.00 & 0.0053 & 1.00 & 0.1017 & 1.00 & 0.1203 & 0.9630 & 0.1234 \\
\multicolumn{1}{l|}{} & $2.5\%$ & $0.6$ & 1.00 & 0.00384 & 1.00 & 0.0516 & 1.00 & 0.0565 & 0.9612 & 0.0498 \\
\multicolumn{1}{l|}{} &  & $1$ & 1.00 & 0.00331 & 1.00 & 0.0371 & 1.00 & 0.0387 & 0.9588 & 0.0382 \\ \cline{2-11} 
\multicolumn{1}{l|}{} &  & $0.2$ & 1.00 & 0.00821 & 1.00 & 0.2243 & 1.00 & 0.2487 & 0.9816 & 0.2425 \\
\multicolumn{1}{l|}{Retirement} & $5\%$ & $0.6$ & 1.00 & 0.0089 & 1.00 & 0.0992 & 1.00 & 0.1072 & 0.9778 & 0.1025 \\
\multicolumn{1}{l|}{} &  & $1$ & 1.00 & 0.00533 & 1.00 & 0.0724 & 1.00 & 0.0781 & 0.9744 & 0.0762 \\ \cline{2-11} 
\multicolumn{1}{l|}{} &  & $0.2$ & 1.00 & 0.0121 & 1.00 & 0.3342 & 1.00 & 0.3637 & 1.00 & 0.3611 \\
\multicolumn{1}{l|}{} & $7.5\%$ & $0.6$ & 1.00 & 0.0185 & 1.00 & 0.1553 & 1.00 & 0.1606 & 1.00 & 0.1587 \\
\multicolumn{1}{l|}{} &  & $1$ & 1.00 & 0.0172 & 1.00 & 0.1098 & 1.00 & 0.1148 & 0.994 & 0.1176 \\ \cline{2-11} 
\multicolumn{1}{l|}{} &  & $0.2$ & 1.00 & 0.0236 & 1.00 & 0.4505 & 1.00 & 0.4938 & 1.00 & 0.4725 \\
\multicolumn{1}{l|}{} & $10\%$ & $0.6$ & 1.00 & 0.0297 & 1.00 & 0.2052 & 1.00 & 0.21005 & 1.00 & 0.2096 \\
\multicolumn{1}{l|}{} & \multicolumn{1}{l}{} & $1$ & 1.00 & 0.03055 & 1.00 & 0.1465 & 1.00 & 0.1519 & 0.998 & 0.1563 \\ \hline
\end{tabular}
\end{table*}

\begin{table*}[!htbp]
\setlength\tabcolsep{1pt}
\centering
\scriptsize
\caption{AUC Values of Attack Detection for CFO. The AUC of MUD is shown in brackets and use ``--" if the AUC value is zero.}
\label{tab:full_auc_table_CFO}
\begin{tabular}{lcccccccccccccccc}
\hline
\multirow{2}{*}{Dataset} & \multirow{2}{*}{$\beta$} & \multirow{2}{*}{$\epsilon$} & \multicolumn{2}{c}{HST-Server} & \multicolumn{2}{c}{HST-User} & \multicolumn{2}{c}{OLH-Server} & \multicolumn{2}{c}{OLH-User} & \multicolumn{2}{c}{OUE} & \multicolumn{2}{c}{GRR} & \multicolumn{2}{c}{OUE-Pad} \\
 &  &  & Ours~/~MUD & \ASG & Ours~/~MUD & \ASG & Ours~/~MUD & \ASG & Ours~/~MUD & \ASG & Ours~/~MUD & \ASG & Ours~/~MUD & \ASG & Ours~/~MUD & \ASG \\ \hline
\multicolumn{1}{l|}{} & \multicolumn{1}{l}{} & $0.2$ & 0.612~/~-- & 0.05 & 0.8928~/~-- & 0.137 & 0.5483~/~-- & 0.036 & 0.9822~/~-- & 0.1565 & 0.9006~/~-- & 0.129 & 1.00~/~-- & 0.493 & 0.866 ~/~-- & 0.0476 \\
\multicolumn{1}{l|}{} & $1\%$ & $0.6$ & 0.57~/~-- & 0.0179 & 0.6816~/~-- & 0.034 & 0.4967~/~-- & 0.0267 & 0.8656~/~-- & 0.0751 & 0.5456~/~-- & 0.0341 & 0.71~/~-- & 0.204 & 0.51 ~/~-- & 0.0297 \\
\multicolumn{1}{l|}{} & \multicolumn{1}{l}{} & $1$ & 0.4322~/~-- & 0.012 & 0.6096~/~-- & 0.021 & 0.4248~/~-- & 0.0181 & 0.8504~/~-- & 0.046 & 0.5352~/~-- & 0.021 & 0.62~/~-- & 0.077 & 0.49 ~/~-- & 0.0206 \\ \cline{2-17} 
\multicolumn{1}{l|}{} &  & $0.2$ & 0.6394~/~-- & 0.10447 & 1.00~/~-- & 0.493 & 0.5528~/~-- & 0.1025 & 1.00~/~-- & 0.3273 & 0.9506~/~-- & 0.493 & 1.00~/~-- & 0.493 & 0.92 ~/~-- & 0.142 \\
\multicolumn{1}{l|}{} & $2.5\%$ & $0.6$ & 0.609~/~-- & 0.04236 & 1.00~/~-- & 0.11669 & 0.5072~/~-- & 0.0565 & 1.00~/~-- & 0.16026 & 0.9422~/~-- & 0.0856 & 0.98~/~-- & 0.46 & 0.89 ~/~-- & 0.061 \\
\multicolumn{1}{l|}{} &  & $1$ & 0.4694~/~-- & 0.0252 & 0.9717~/~-- & 0.066 & 0.4667~/~-- & 0.041 & 0.92~/~-- & 0.109 & 0.5933~/~-- & 0.053 & 0.76~/~-- & 0.224 & 0.54 ~/~-- & 0.043 \\ \cline{2-17} 
\multicolumn{1}{l|}{} &  & $0.2$ & 0.7128~/~-- & 0.222 & 1.00~/~-- & 0.493 & 0.7192~/~-- & 0.243 & 1.00~/~-- & 0.4395 & 1.00~/~-- & 0.493 & 1.00~/~-- & 0.493 & 1.00 ~/~-- & 0.31 \\
\multicolumn{1}{l|}{$N(0, 10)$} & $5\%$ & $0.6$ & 0.6912~/~-- & 0.077 & 1.00~/~-- & 0.493 & 0.6336~/~-- & 0.106 & 1.00~/~-- & 0.283 & 1.00~/~-- & 0.493 & 1.00~/~-- & 0.493 & 1.00 ~/~-- & 0.118 \\
\multicolumn{1}{l|}{} &  & $1$ & 0.5912~/~-- & 0.051 & 0.9784~/~-- & 0.47 & 0.544~/~-- & 0.078 & 0.9272~/~-- & 0.1897 & 1.00~/~-- & 0.115 & 1.00~/~-- & 0.45 & 1.00 ~/~-- & 0.085 \\ \cline{2-17}
\multicolumn{1}{l|}{} &  & $0.2$ & 0.7706~/~-- & 0.341 & 0.8822~/~-- & 0.493 & 0.7489~/~-- & 0.342 & 1.00~/~-- & 0.457 & 1.00~/~-- & 0.493 & 1.00~/~-- & 0.493 & 1.00 ~/~-- & 0.433 \\
\multicolumn{1}{l|}{} & $7.5\%$ & $0.6$ & 0.7256~/~-- & 0.117 & 1.00~/~-- & 0.493 & 0.8611~/~-- & 0.158 & 1.00~/~-- & 0.3718 & 1.00~/~-- & 0.493 & 1.00~/~-- & 0.493 & 1.00 ~/~-- & 0.178 \\
\multicolumn{1}{l|}{} &  & $1$ & 0.689~/~-- & 0.0795 & 1.00~/~-- & 0.493 & 0.867~/~-- & 0.117 & 0.9312~/~-- & 0.268 & 1.00~/~-- & 0.211 & 1.00~/~-- & 0.493 & 1.00 ~/~-- & 0.127 \\ \cline{2-17}
\multicolumn{1}{l|}{} &  & $0.2$ & 0.9416~/~0.525 & 0.453 & 1.00~/~0.525 & 0.493 & 0.9928~/~0.525 & 0.448 & 1.00~/~0.525 & 0.472 & 1.00~/~0.525 & 0.493 & 1.00~/~-- & 0.493 & 1.00~/~0.525 & 0.493 \\
\multicolumn{1}{l|}{} & $10\%$ & $0.6$ & 0.8944~/~-- & 0.162 & 1.00~/~-- & 0.493 & 0.9208~/~-- & 0.214 & 1.00~/~-- & 0.435 & 1.00~/~-- & 0.493 & 1.00~/~-- & 0.493 & 1.00 ~/~-- & 0.237 \\
\multicolumn{1}{l|}{} & \multicolumn{1}{l}{} & $1$ & 0.7088~/~-- & 0.105 & 1.00~/~-- & 0.493 & 0.9~/~-- & 0.158 & 0.9661~/~-- & 0.331 & 1.00~/~-- & 0.493 & 1.00~/~-- & 0.493 & 1.00 ~/~-- & 0.167 \\ \hline

\multicolumn{1}{l|}{} & \multicolumn{1}{l}{} & $0.2$ & 0.4416~/~-- & 0.0484 & 1.00~/~-- & 0.391 & 0.476~/~-- & 0.047 & 0.9952~/~-- & 0.207 & 1.00~/~-- & 0.391 & 1.00~/~-- & 0.391 & 1.00~/~-- & 0.094  \\
\multicolumn{1}{l|}{} & $1\%$ & $0.6$ & 0.4384~/~-- & 0.0054 & 1.00~/~-- & 0.1014 & 0.4232~/~-- & 0.0189 & 0.6224~/~-- & 0.105 & 1.00~/~-- & 0.071 & 1.00~/~-- & 0.184 & 1.00~/~-- & 0.021 \\
\multicolumn{1}{l|}{} & \multicolumn{1}{l}{} & $1$ & 0.3306~/~-- & 0.0082 & 0.9744~/~-- & 0.058 & 0.392~/~-- & 0.015 & 0.5784~/~-- & 0.0635 & 0.6808~/~-- & 0.0261 & 0.8992~/~-- & 0.096 & 0.6407~/~-- & 0.016  \\ \cline{2-17} 
\multicolumn{1}{l|}{} &  & $0.2$ & 0.4972~/~-- & 0.0992 & 1.00~/~-- & 0.391 & 0.6167~/~-- & 0.117 & 1.00~/~-- & 0.285 & 1.00~/~-- & 0.391 & 1.00~/~-- & 0.391 & 1.00~/~-- & 0.195 \\
\multicolumn{1}{l|}{} & $2.5\%$ & $0.6$ & 0.4872~/~-- & 0.0343 & 1.00~/~-- & 0.392 & 0.4933~/~-- & 0.041 & 1.00~/~-- & 0.201 & 1.00~/~-- & 0.382 & 1.00~/~-- & 0.387 & 1.00~/~-- & 0.063 \\
\multicolumn{1}{l|}{} &  & $1$ & 0.3696~/~-- & 0.027 & 1.00~/~-- & 0.344 & 0.4~/~-- & 0.03518 & 1.00~/~-- & 0.146 & 0.9696~/~-- & 0.11 & 1.00~/~-- & 0.22 & 0.913~/~-- & 0.044 \\ \cline{2-17} 
\multicolumn{1}{l|}{} &  & $0.2$ & 0.555~/~-- & 0.188 & 1.00~/~-- & 0.39 & 0.6504~/~-- & 0.204 & 1.00~/~-- & 0.335 & 1.00~/~-- & 0.391 & 1.00~/~-- & 0.391 & 1.00~/~-- & 0.391 \\
\multicolumn{1}{l|}{Taxi} & $5\%$ & $0.6$ & 0.5352~/~-- & 0.072 & 1.00~/~-- & 0.391 & 0.5392~/~-- & 0.1069 & 1.00~/~-- & 0.277 & 1.00~/~-- & 0.391 & 1.00~/~-- & 0.391 & 1.00~/~-- & 0.147 \\
\multicolumn{1}{l|}{} &  & $1$ & 0.4976~/~-- & 0.0486 & 1.00~/~-- & 0.389 & 0.4844~/~-- & 0.077 & 1.00~/~-- & 0.23 & 1.00~/~-- & 0.389 & 1.00~/~-- & 0.385 & 1.00~/~-- & 0.102 \\ \cline{2-17}
\multicolumn{1}{l|}{} &  & $0.2$ & 0.6211~/~-- & 0.296 & 1.00~/~-- & 0.391 & 0.7411~/~-- & 0.31 & 1.00~/~-- & 0.357 & 1.00~/~-- & 0.392 & 1.00~/~-- & 0.391 & 1.00~/~-- & 0.391 \\
\multicolumn{1}{l|}{} & $7.5\%$ & $0.6$ & 0.6089~/~-- & 0.108 & 1.00~/~-- & 0.392 & 0.6494~/~-- & 0.156 & 1.00~/~-- & 0.312 & 1.00~/~-- & 0.391 & 1.00~/~-- & 0.391 & 1.00~/~-- & 0.255 \\
\multicolumn{1}{l|}{} &  & $1$ & 0.5667~/~-- & 0.073 & 1.00~/~-- & 0.39105 & 0.5583~/~-- & 0.1161 & 1.00~/~-- & 0.276 & 1.00~/~-- & 0.391 & 1.00~/~-- & 0.391 & 1.00~/~-- & 0.154 \\ \cline{2-17}
\multicolumn{1}{l|}{} &  & $0.2$ & 0.8432~/~0.575 & 0.382 & 1.00~/~0.575 & 0.392 & 0.7912~/~0.55 & 0.377 & 1.00~/~0.55 & 0.37 & 1.00~/~0.55 & 0.392 & 1.00~/~-- & 0.391 & 1.00~/0.55 & 0.391 \\
\multicolumn{1}{l|}{} & $10\%$ & $0.6$ & 0.62~/~-- & 0.15 & 1.00~/~-- & 0.39 & 0.63~/~-- & 0.207 & 1.00~/~-- & 0.335 & 1.00~/~-- & 0.391 & 1.00~/~-- & 0.391 & 1.00~/~-- & 0.391 \\
\multicolumn{1}{l|}{} & \multicolumn{1}{l}{} & $1$ & 0.5072~/~-- & 0.0976 & 1.00~/~-- & 0.3913 & 0.5338~/~-- & 0.151 & 1.00~/~-- & 0.303 & 1.00~/~-- & 0.391 & 1.00~/~-- & 0.391 & 1.00~/~-- & 0.248 \\ \hline

\multicolumn{1}{l|}{} & \multicolumn{1}{l}{} & $0.2$ & 0.53~/~-- & 0.1473 & 0.705~/~-- & 0.1836 & 0.505~/~-- & 0.138 & 0.61~/~-- & 0.356 & 0.7~/~-- & 0.228 & 1.00~/~-- & 0.731 & 0.65~/~-- & 0.123  \\
\multicolumn{1}{l|}{} & $1\%$ & $0.6$ & 0.475~/~-- & 0.0591 & 0.68~/~-- & 0.04072 & 0.42~/~-- & 0.0452 & 0.52~/~-- & 0.1324 & 0.67~/~-- & 0.051 & 0.81~/~-- & 0.262 & 0.57~/~-- & 0.04 \\
\multicolumn{1}{l|}{} & \multicolumn{1}{l}{} & $1$ & 0.465~/~-- & 0.0262 & 0.59~/~-- & 0.0067 & 0.345~/~-- & 0.025 & 0.405~/~-- & 0.0763 & 0.37~/~-- & 0.0269 & 0.56~/~-- & 0.128 & 0.3~/~-- & 0.023  \\ \cline{2-17} 
\multicolumn{1}{l|}{} &  & $0.2$ & 0.5506~/~-- & 0.204 & 1.00~/~-- & 0.728 & 0.53~/~-- & 0.2679 & 1.00~/~-- & 0.5875 & 0.9089~/~-- & 0.71 & 1.00~/~-- & 0.731 & 0.864~/~-- & 0.23\\
\multicolumn{1}{l|}{} & $2.5\%$ & $0.6$ & 0.535~/~-- & 0.091 & 0.9978~/~-- & 0.281 & 0.52~/~-- & 0.101 & 1.00~/~-- & 0.293 & 0.84~/~-- & 0.127 & 0.91~/~-- & 0.66 & 0.77~/~-- & 0.083 \\
\multicolumn{1}{l|}{} &  & $1$ & 0.5194~/~-- & 0.043 & 0.8733~/~-- & 0.078 & 0.4933~/~-- & 0.0693 & 0.9833~/~-- & 0.168 & 0.6572~/~-- & 0.059 & 0.831~/~-- & 0.332 & 0.603~/~-- & 0.056 \\ \cline{2-17} 
\multicolumn{1}{l|}{} &  & $0.2$ & 0.58~/~-- & 0.358 & 1.00~/~-- & 0.7284 & 0.5906~/~-- & 0.35 & 1.00~/~-- & 0.6694 & 1.00~/~-- & 0.7284 & 1.00~/~-- & 0.731 & 1.00~/~-- & 0.456 \\
\multicolumn{1}{l|}{Retirement} & $5\%$ & $0.6$ & 0.575~/~-- & 0.128 & 1.00~/~-- & 0.726 & 0.565~/~-- & 0.166 & 1.00~/~-- & 0.5842 & 1.00~/~-- & 0.695 & 1.00~/~-- & 0.731 & 1.00~/~-- & 0.151 \\
\multicolumn{1}{l|}{} &  & $1$ & 0.5574~/~-- & 0.0867 & 1.00~/~-- & 0.684 & 0.5402~/~-- & 0.1257 & 1.00~/~-- & 0.347 & 0.9914~/~-- & 0.1641 & 1.00~/~-- & 0.67 & 0.99~/~-- & 0.11 \\ \cline{2-17}
\multicolumn{1}{l|}{} &  & $0.2$ & 0.66~/~-- & 0.53 & 1.00~/~-- & 0.731 & 0.7372~/~-- & 0.527 & 1.00~/~-- & 0.6977 & 1.00~/~-- & 0.731 & 1.00~/~-- & 0.731 & 1.00~/~-- & 0.693 \\
\multicolumn{1}{l|}{} & $7.5\%$ & $0.6$ & 0.6456~/~-- & 0.172 & 1.00~/~-- & 0.731 & 0.63~/~-- & 0.2277 & 1.00~/~-- & 0.656 & 1.00~/~-- & 0.731 & 1.00~/~-- & 0.731 & 1.00~/~-- & 0.248 \\
\multicolumn{1}{l|}{} &  & $1$ & 0.5839~/~-- & 0.1091 & 1.00~/~-- & 0.731 & 0.6656~/~-- & 0.1686 & 1.00~/~-- & 0.5261 & 1.00~/~-- & 0.3577 & 1.00~/~-- & 0.731 & 1.00~/~-- & 0.166 \\ \cline{2-17}
\multicolumn{1}{l|}{} &  & $0.2$ & 0.725~/~0.525 & 0.6786 & 1.00~/~0.525 & 0.7326 & 0.88~/~0.525 & 0.6582 & 1.00~/~0.525 & 0.7142 & 1.00~/~0.525 & 0.732 & 1.00~/~-- & 0.731 & 1.00~/~0.525 & 0.731 \\
\multicolumn{1}{l|}{} & $10\%$ & $0.6$ & 0.71~/~-- & 0.2405 & 1.00~/~-- & 0.732 & 0.8292~/~-- & 0.3216 & 1.00~/~-- & 0.675 & 1.00~/~-- & 0.732 & 1.00~/~-- & 0.731 & 1.00~/~-- & 0.337 \\
\multicolumn{1}{l|}{} & \multicolumn{1}{l}{} & $1$ & 0.5976~/~-- & 0.1588 & 1.00~/~-- & 0.732 & 0.6783~/~-- & 0.231 & 1.00~/~-- & 0.6574 & 1.00~/~-- & 0.729 & 1.00~/~-- & 0.731 & 1.00~/~-- & 0.233 \\ \hline
\end{tabular}
\end{table*}

\end{appendices}

\end{document}